\documentclass{article}

\usepackage{amsmath, amssymb, amsfonts}
\usepackage{natbib}
\usepackage{amsthm}
\usepackage[english]{babel}
\usepackage{color}

\usepackage{geometry}
\usepackage[varg]{txfonts}
\newtheorem{theorem}{Theorem}
\newtheorem{definition}{Definition}
\newtheorem{corollary}{Corollary}
\newtheorem{proposition}{Proposition}
\newtheorem{lemma}{Lemma}
\theoremstyle{remark}
\newtheorem{remark}{Remark}

\begin{document}

\raggedbottom
\allowdisplaybreaks

\title{\bfseries Sequential Posted Price Mechanisms with Correlated Valuations}
\author{Marek Adamczyk \\ \small{\texttt{adamczyk@dis.uniroma1.it}}  \\ \small{Sapienza University of Rome} \and Allan Borodin \\ \small{\texttt{bor@cs.toronto.edu}} \\ \small{University of Toronto} \and Diodato Ferraioli \\ \small{\texttt{dferraioli@unisa.it}} \\ \small{University of Salerno} \and Bart de Keijzer \\ \small{\texttt{dekeijzer@dis.uniroma1.it}} \\ \small{Sapienza University of Rome} \and Stefano Leonardi \\ \small{\texttt{leonardi@dis.uniroma1.it}} \\ \small{Sapienza University of Rome}}
\sloppy 
\date{}

\maketitle

\begin{abstract}
We study the revenue performance of sequential posted price mechanisms and some natural extensions, for a general setting where the valuations of the buyers are drawn from a correlated distribution. Sequential posted price mechanisms are conceptually simple mechanisms that work by proposing a ``take-it-or-leave-it'' offer to each buyer. We apply sequential posted price mechanisms to single-parameter multi-unit settings in which each buyer demands only one item and the mechanism can assign the service to at most $k$ of the buyers. For standard sequential posted price mechanisms, we prove that with the valuation distribution having finite support, no sequential posted price mechanism can extract a constant fraction of the optimal expected revenue, even with unlimited supply. We extend this result to the the case of a continuous valuation distribution when various standard assumptions hold simultaneously (i.e., everywhere-supported, continuous, symmetric, and normalized (conditional) distributions that satisfy \emph{regularity}, the \emph{MHR condition}, and \emph{affiliation}). In fact, it turns out that the best fraction of the optimal revenue that is extractable by a sequential posted price mechanism is proportional to ratio of the highest and lowest possible valuation. We prove that for two simple generalizations of these mechanisms, a better revenue performance can be achieved: if the sequential posted price mechanism has for each buyer the option of \emph{either} proposing an offer \emph{or} asking the buyer for its valuation, then a $\Omega(1/\max\{1,d\})$ fraction of the optimal revenue can be extracted, where $d$ denotes the degree of dependence of the valuations, ranging from complete independence ($d=0$) to arbitrary dependence ($d = n-1$). Moreover, when we generalize the sequential posted price mechanisms further, such that the mechanism has the ability to make a take-it-or-leave-it offer to the $i^{\text{th}}$ buyer that depends on the valuations of all buyers except $i$'s, we prove that a constant fraction $(2 - \sqrt{e})/4 \approx 0.088$ of the optimal revenue can be always extracted. 
\end{abstract}

\section{Introduction}
A large body of literature in the field of mechanism design focuses on the design of auctions that are optimal with respect some given objective function,
such as maximizing the social welfare or the auctioneer's revenue. This literature mainly considered direct revelation mechanisms, in which each buyer submits a bid that represents his valuation for getting the service, and the mechanism determines the winners and how much they are forced to pay. The reason for this is the \emph{revelation principle} (see, e.g., \citep{boergersmechdesign}), which implies that one may resort to studying only direct revelation mechanisms for many purposes, such as maximizing the social welfare or the auctioneer's revenue. Some of the most celebrated mechanisms follow this approach,
such as the Vickrey-Clark-Groves mechanism \citep{vickrey,clarke,groves} and the Myerson mechanism \citep{myerson}.

A natural assumption behind these mechanisms is that buyers will submit truthfully
whenever the utility they take with the truthful bid is at least as high as
the utility they may take with a different bid.
However, it has often been acknowledged that such an assumption may be too strong in a real world setting.
In particular, \citet{sandholmgilpin} highlight that this assumption usually fails because of:
1) a buyer's unwillingness to fully specify their values,
2) a buyer's unwillingness to participate in ill understood, complex, unintuitive auction mechanisms,
and 3) irrationality of a buyer, which leads him to underbid even when it is known that there is nothing to be gained from this behavior.

This has recently motivated the research about auction mechanisms that are conceptually simple. 
Among these, the class of \emph{sequential posted price mechanisms} \citep{chawla2010} is particularly attractive.
First studied by \citet{sandholmgilpin} (and called ``take-it-or-leave-it mechanisms''), these mechanisms 
work by iteratively selecting a buyer that has not been selected previously,
and offering him a price.
The buyer may then accept or reject that price.
When the buyer accepts, he is allocated the service.
Otherwise, the mechanism does not allocate the service to the buyer.
In the sequential posted-price mechanism we allow both the choice of buyer
and the price offered to that buyer to depend on the decisions of the previously selected buyers
(and the prior knowledge about the buyers' valuations).
Also, randomization in the choice of the buyer and in the charged price is allowed.
Sequential posted price mechanisms are thus conceptually simple and buyers do not have to reveal their valuations for the service.
Moreover, they possess a trivial dominant strategy (i.e., do not require 
any strategic decisions on the part of a buyer) and are individually rational (i.e., informally, participation in such an auction is never harmful to the buyer).

Sequential posted price mechanisms have been mainly studied for the setting
where the valuations of the buyers are each drawn independently from publicly known buyer-specific distributions,
called the \emph{independent values} setting.
In this paper, we study a much more general setting, and assume that the entire vector of valuations is drawn
from one publicly known distribution, which allows for arbitrarily complex dependencies among the valuations of the buyers.
This setting is commonly known as the \emph{correlated values} setting.
Our goal is to investigate questions related to the existence of sequential posted price mechanisms that achieve a high revenue.
That is, we quantify the quality of a mechanism by comparing its expected revenue to that of the \emph{optimal mechanism}: the mechanism that achieves the highest expected revenue among all dominant strategy incentive compatible and ex-post individually rational mechanisms (see the definitions in Section \ref{sec:preliminaries} below).

We assume a standard Bayesian, transferable, quasi-linear utility model and
we study a \emph{unit demand, single parameter, multi-unit} setting:
there is one service (or type of item) being provided by the auctioneer, any buyer is interested in receiving the service once, 
and the \emph{valuation} of each buyer consists of a single number
that reflects to what extent a buyer would profit from receiving the service provided by the auctioneer. The auctioneer can charge a price to a bidder, so that the utility of a bidder is his valuation (in case he gets the service), minus the charged price.
We focus in this paper on the $k$-limited supply setting, where service can be provided to at most $k$ of the buyers.
This is an important setting because it is a natural constraint in many realistic scenarios,
and it contains two fundamental special cases: the \emph{unit supply} setting (where $k = 1$), and the \emph{unlimited supply} setting where $k = n$.

\paragraph{Related work}
There has been recent substantial work on the subject of revenue performance for \emph{simple} mechanisms \citep{HartlineR09, HartN12, BabaioffILW14, RubinsteinW15,DevanurMSW15}.
In particular, \citet{BabaioffILW14} highlight the importance of understanding
what is the relative strength of simple versus complex mechanisms with regard to revenue maximization.

As described above, sequential posted price mechanisms are an example of such a simple class of mechanisms.
\citet{sandholmgilpin} have been the first ones to study sequential posted price mechanisms.
They give experimental results for the case in which values are independently drawn from the uniform distribution in $[0,1]$.
Moreover, they consider the case where multiple offers can be made to a bidder, and study the equilibria that arise from this.
\citet{blumrosenholenstein} compare fixed price (called symmetric auctions),
sequential posted price (called discriminatory auctions) and the optimal mechanism for valuations drawn from a wide class of i.i.d distributions.
\citet{BabaioffDKS12} consider \emph{prior-independent} posted price mechanisms with $k$-limited supply for the setting where
the only information known about the valuation distribution is that all valuations are independently drawn from the same distribution with support $[0,1]$. 
Posted-price mechanisms have also been previously studied in \citep{kleinbergleighton,blumhartline,blumkumar}, albeit for a non-Bayesian, on-line setting.
In a recent work \cite{feldmangravinlucier} study ``on-line'' posted price mechanisms for combinatorial auctions when valuations are independently drawn.

The works of \citet{chawla2010} and \citet{guptanagarajan} are closest to our present work,
although they only consider sequential posted price mechanisms in the independent values setting.
In particular, \citet{chawla2010} prove that such mechanisms can extract a constant factor of the optimal revenue 
for single and multiple parameter settings under various constraints on the allocations.
They also consider \emph{order-oblivious} (i.e., ``on-line'') sequential posted price mechanisms in which the order in which the order of the buyers is fixed and adversarially determined. They use order-oblivious mechanisms in order to establish some results for the more general multi-parameter case.
\citet{yan} builds on this work and strengthens some of the results of \citet{chawla2010}. Moreover, \citet{kleinbergweinberg} prove results that imply a strengthening of some of the results of \citet{chawla2010}. 

\citet{guptanagarajan} introduce a more abstract stochastic probing problem
that includes Bayesian sequential posted price mechanisms as well as
the stochastic matching problem introduced by~\citet{ChenIKKR09}. Their approximation bounds were later improved
by \citet{AdamczykSW14} who in particular matched the approximation of \citet{chawla2010} for single matroid settings. 

All previous work only consider the independent setting.
In this work we instead focus on the correlated setting.
The lookahead mechanism of \cite{ronen} is a fundamental reference for the
correlated setting and resembles our blind auction mechanism but is
different in substantial ways as we will soon indicate.
\citet{CremMcLe} made a fundamental contribution to auction theory in the correlated value setting, by characterizing exactly for which valuation distributions it is possible to extract the full optimal social welfare as revenue. They do this for the ex-post IC, interim IR mechanisms and for the dominant strategy IC, interim IR mechanisms.
%
\citet{segal} give a characterization of optimal ex-post incentive compatible and ex-post individually rational optimal mechanisms.
\citet{roughgarden-talgamcohen} study optimal mechanism design in the even more general \emph{interdependent} setting.
They show how to extend the Myerson mechanism to this setting for various assumptions on the valuation distribution.
There is now a substantial  literature \citep{dobzinski2011, roughgarden-talgamcohen, chawla2014approximate},
that develop mechanisms with good approximation guarantees for revenue maximization in the correlated setting.
These mechanisms build on the lookahead mechanism of \cite{ronen}
and thus they also differ from the mechanisms proposed in this work.

\paragraph{Contributions and outline}
We define some preliminaries and notation in Section \ref{sec:preliminaries}. In Section \ref{sec:nonexistence} we give a simple sequence of instances which demonstrates that 
for unrestricted correlated distributions the expected revenue of the best sequential posted price mechanism does not approximate within any constant factor the expected revenue of the optimal dominant strategy incentive compatible and ex-post individually rational mechanism. This holds for any value of $k$ (i.e, the size of the supply). We extend this impossibility result by proving that a constant approximation is impossible to achieve even when we assume that the valuation distribution is continuous and satisfies all of the following conditions simultaneously: the valuation distribution is supported everywhere, is entirely symmetric, satisfies \emph{regularity}, satisfies the \emph{monotone hazard rate} condition, satisfies \emph{affiliation}, all the induced marginal distributions have finite expectation, and all the conditional marginal distributions are non-zero everywhere.

The maximum revenue that a sequential posted price mechanism can generate on our examples is shown to be characterized by the logarithm of the ratio between the highest and lowest valuations in the support of the distribution. We show in Section \ref{sec:sppmpositive} that this approximation ratio is essentially tight.

Given these negative results, we consider a generalization of sequential posted price mechanisms that are more suitable for settings with limited dependence among the buyers' valuations: \emph{enhanced sequential posted price mechanisms}.
An enhanced sequential posted price mechanism
works by iteratively selecting a buyer that has not been selected previously.
The auctioneer can either offer the selected buyer a price or ask him to report his valuation.
As in sequential posted price mechanisms, if the buyer is offered a price, then he may accept or reject that price.
When the buyer accepts, he is allocated the service.
Otherwise, the mechanism does not allocate the service to the buyer.
If instead, the buyer is asked to report his valuation, then the mechanism does not allocate him the service.
Note that the enhanced sequential posted price mechanism requires that some 
fraction of buyers reveal their valuation truthfully.
Thus, the original property that the bidders not have to reveal their preferences is \emph{partially} sacrificed, in return for a more powerful class of mechanisms and (as we will see) a better revenue performance.
For practical implementation such mechanisms can be slightly adjusted by providing a bidder with a small monetary reward in case he is asked to reveal his valuation\footnote{We also note that in some realistic scenarios, the valuation of some buyers may be known a priori to the auctioneer
(through, for example, the repetition of auctions or accounting and profiling operations), which can be exploited accordingly in the design of the enhanced posted price mechanism.}.

For the enhanced sequential posted price mechanisms, we prove that again there are instances in which the revenue is not within a constant fraction of the optimal revenue. However, we show that this class of mechanisms can extract a fraction $\Theta(1/n)$ of the optimal revenue, i.e., a fraction that is independent of the valuation distribution.

This result seems to suggest that to achieve a constant approximation of the optimal revenue it is \emph{necessary} to collect all the bids truthfully.
Consistent with this hypothesis, we prove that a constant fraction of the optimal revenue can be extracted by a dominant strategy IC \emph{blind offer mechanisms}: these mechanisms inherit all the limitations of sequential posted price mechanisms (i.e., buyers are considered sequentially in an order independent of the bids; it is possible to offer a price to a buyer only when selected; and the buyer gets the service only if it accepts the offered price), except that the price offered to a bidder $i$ may now depend on the bids submitted by all players other than $i$. 
This generalization sacrifices entirely the property that buyers valuations need not be revealed, and the class of blind offer mechanisms are thus necessarily direct revelation mechanisms. However, this comes with the reward of a revenue that is only a constant factor away from optimal. 
Also, blind offer mechanisms preserve the conceptual simplicity of sequential posted price mechanisms, and are easy to grasp for the buyers participating in the auction: in particular, buyers have a conceptually simple and practical strategy, to accept the price if and only if it is not above their valuation,
regardless of how the prices are computed.
Unfortunately it turns out that blind offer mechanisms are not inherently truthful, although this is only for a subtle reason: while it is true that a buyer cannot influence the price he is offered nor the decision of the mechanism to pick him, he can still be incentivized to misreport in order to influence the probability that the supply has not run out before the buyer is picked. However, this is not a big obstruction, as we will show that there is a straightforward way to turn any blind offer mechanism into an incentive compatible one.
We stress that, even if blind offer mechanisms sacrifice some simplicity (and practicality),
we still find it theoretically interesting that a mechanism that on-line allocates items to buyers ({i.e., {\it in any order} )
and thus not necessarily allocating the items to agents maximizing their profit,say as in \cite{ronen} and  \cite{chawla2010})
is able to achieve a constant approximation of the optimal revenue even with correlated valuations. 
Moreover, our result for blind offer mechanisms has a constructive purpose: 
it provides the intermediate step en route to 
establishing revenue approximation bounds for other mechanisms.
We will show how blind offer mechanisms serve to this purpose in Section~\ref{sec:generalization} for the construction of enhanced sequential posted 
price mechanisms. 

We highlight that our positive results do not make any assumptions on the marginal valuation distributions of the buyers nor the type of correlation among the buyers. However, in Section \ref{sec:generalization} we consider the case in which the degree of dependence among the buyers is limited. In particular, we introduce the notion of \emph{$d$-dimensionally dependent distributions}. This notion informally requires that for each buyer $i$ there is a set $S_i$ of $d$ other buyers such that the distribution of $i$'s valuation when conditioning on the vector of other buyers' valuations can likewise be obtained by only conditioning on the valuations of $S_i$. Thus, this notion induces a hierarchy of $n$ classes of valuation distributions with increasing degrees of dependence among the buyers: for $d = 0$ the buyers have independent valuations, while the other extreme $d = n-1$ implies that the valuations may be dependent in arbitrarily complex ways. Note that $d$-dimensional dependence does not require that the marginal valuation distributions of the buyers themselves satisfy any particular property, and neither does it require anything from the type of correlation that may exist among the buyers. This stands in contrast with commonly made assumptions such as \emph{symmetry}, \emph{affiliation}, the \emph{monotone-hazard rate assumption}, and \emph{regularity}, that are often encountered in the auction theory and mechanism design literature.

Our main positive result for enhanced sequential posted price mechanisms then states that when the valuation distribution is $d$-dimensionally dependent, there exists an enhanced sequential posted price mechanism that extracts an $\Omega(1/d)$ fraction of the optimal revenue. The proof of this result consists of three key ingredients:
\begin{itemize}
\item An upper bound on the optimal ex-post IC, ex-post IR revenue in terms of the solution of a linear program. This upper bound has the form of a relatively simple expression that is important for the definition and analysis of a blind offer mechanism that we define subsequently. This part of the proof generalizes a linear programming characterization introduced by \citet{guptanagarajan} for 
the independent distribution setting.
\item A proof for the fact that incentive compatible blind offer mechanisms are powerful enough to extract a constant fraction of the optimal revenue of any instance. This makes crucial use of the linear program mentioned above. 
\item A conversion lemma showing that blind offer mechanisms can be turned into enhanced sequential posted price mechanisms while maintaining a fraction $\Omega(1/d)$ of the revenue of the blind offer mechanism\footnote{The sequential posted price mechanism constructed in this lemma depends on a parameter $q$ which can be adjusted to trade off the amount of valuation elicitation against the hidden constant in the $\Omega(1/d)$ expression.}.
\end{itemize}

While our focus is on proving the \emph{(non-)existence} of simple mechanisms that perform well in terms of revenue, we note the following about the computational complexity of our mechanisms: all the mechanisms that we use in  our positive results run in polynomial time when the valuation distribution is given as a description of the valuation vectors together with their probability mass\footnote{When the valuation distribution is not accessible in such a form, and can instead only be sampled from, then standard sampling techniques can be used in order to obtain an estimate of the distribution. When this estimate is reasonably accurate (which should be the case when the distribution in question does not have extreme outliers), then the mechanisms in our paper can still be used with only a small additional loss in revenue.}.

Additionally, we note that all of our negative results hold for randomized mechanisms. On the other hand: our positive results only require randomization in a limited way. For our positive results for classical sequential posted price mechanisms, only the offered prices need to be randomly chosen, while the order in which the agents are picked is arbitrary. This makes these positive results hold through for the \emph{order-oblivious} (i.e., on-line) setting in which the mechanism has no control over the agent that is picked in each iteration. Our positive result for blind offer mechanisms only requires randomized pricing in case $k < n$ and works for any ordering in which the agents are picked, as long as the mechanisms knows the ordering in advance. Our positive result for enhanced sequential posted price mechanisms requires randomized pricing and the assumption that the mechanism can pick a uniformly random ordering of the agents (i.e., holds
in the \emph{random order model} ROM of arrivals). 

Some of the proofs have been omitted in the main body of the paper. In some occasions we have replaced them by proof sketches. In all of these cases the full versions of the proofs can be found in Appendix \ref{sec:missingproofs}.

\section{Preliminaries and notation}\label{sec:preliminaries}
For $a \in \mathbb{N}$, we write $[a]$ to denote the set $\{1,\ldots,a\}$. We write $\mathbf{1}[X]$ to denote the indicator function for property $X$ (i.e., it evaluates to $1$ if $X$ holds, and to $0$ otherwise). When $\vec{v}$ is a vector and $a$ is an arbitrary element, we denote by $(a,\vec{v}_{-i})$ the vector obtained by replacing $v_i$ with $a$.

We face a setting where an auctioneer provides a service to $n$ buyers, and is able to serve at most $k$ of the buyers. As mentioned in the introduction, the buyers have valuations for the service offered, which are drawn from a \emph{valuation distribution}, defined as follows.
\begin{definition}[Valuation distribution]
A \emph{valuation distribution} $\pi$ for $n$ buyers is a probability distribution on $\mathbb{R}_{\geq 0}^n$. 
\end{definition}
We will assume throughout this paper that $\pi$ is discrete, except for in Theorem~\ref{prop:spp_continuous}.
In that proposition we assume that some standard assumptions about continuous valuation distributions hold, such as \emph{regularity}, the \emph{monotone hazard rate (MHR)} condition and \emph{affiliation}. We refer the interested reader to Appendix~\ref{apx:cont_prop} for a definition and a brief discussion of these properties.

We will use the following notation for conditional and marginal probability distributions.
For an arbitrary probability distribution $\pi$, denote by $\text{supp}(\pi)$ the support of $\pi$.
Let $\pi$ be a discrete finite probability distribution on $\mathbb{R}^n$,
let $i \in [n]$, $S \subset [n]$ and $\vec{v} \in \mathbb{R}^{n}$.
We denote by $\vec{v}_S$ the vector obtained by removing from $\vec{v}$ the coordinates in $[n] \setminus S$.
We denote by $\pi_S$ the probability distribution induced by drawing a vector from $\pi$
and removing the coordinates corresponding to index set $[n] \setminus S$.
If $S = \{i\}$ is a singleton, we write $\pi_i$ instead of $\pi_{\{i\}}$
and if $S$ consists of all but one buyer $i$, we write $\pi_{-i}$ instead of $\pi_{[n]\setminus\{i\}}$.
We denote by $\pi_{\vec{v}_S}$ the probability distribution of $\pi$ conditioned on the event that
$\vec{v}_S$ is the vector of values on the coordinates corresponding to index set $S$.
We denote by $\pi_{i, \vec{v}_S}$ the marginal probability distribution of the coordinate of $\pi_{\vec{v}_S}$
that corresponds to buyer $i$. Again, in these cases, in the subscript we will also write $i$
instead of $\{i\}$ and $-i$ instead of $[n]\setminus\{i\}$.

Each of the buyers is interested in receiving the service at most once.
The auctioneer runs a \emph{mechanism} with which the buyers interact.
In general, a mechanism consists of a specification of (i.) the strategies available to the buyers,
(ii.) a function that maps each vector of strategies chosen by the buyers to an outcome.
The mechanism, when provided with a strategy profile of the buyers,
outputs an outcome that consists of a vector $\vec{x} = (x_1, \ldots, x_n)$ and a vector $\vec{p} = (p_1, \ldots, p_n)$:
vector $\vec{x}$ is the \emph{allocation} vector, i.e., the $(0,1)$-vector that indicates to which of the buyers the auctioneer allocates the service,
and $\vec{p} = (p_1, \ldots, p_n)$ is the vector of \emph{prices} that the auctioneer asks from the buyers.
For outcome $(\vec{x},\vec{p})$, the utility of a buyer $i \in [n]$ is $x_i v_i - p_i$.
The auctioneer is interested in maximizing the \emph{revenue} $\sum_{i \in [n]} p_i$,
and is assumed to have full knowledge of the valuation distribution, but not of the actual valuations of the buyers.

We formalize the above as follows.
\begin{definition}
An \emph{instance} is a triple $(n,\pi,k)$, where $n$ is the number of participating buyers, $\pi$ is the valuation distribution, and $k \in \mathbb{N}_{\geq 1}$ is the amount of services that the auctioneer may allocate to the buyers. A \emph{deterministic mechanism} $f$ is a function from $\times_{i \in [n]} \Sigma_i$ to $\{0,1\}^n \times \mathbb{R}_{\geq 0}^n$, for any choice of \emph{strategy sets} $\Sigma_i,i\in [n]$. When $\Sigma_i = \text{supp}(\pi_i)$ for all $i \in [n]$, mechanism $f$ is called a deterministic \emph{direct revelation mechanism}. A \emph{randomized mechanism} $M$ is a probability distribution over deterministic mechanisms. For $i \in [n]$ and $\vec{s} \in \times_{j \in [n]} \Sigma_j$, we will denote $i$'s \emph{expected allocation} $\mathbf{E}_{f \sim M}[f(\vec{s})_i]$ by $x_i(\vec{s})$ and $i$'s \emph{expected payment} $\mathbf{E}_{f \sim M}[f(\vec{s})_{n+i}]$ by $p_i(\vec{s})$. (When we use this notation, the mechanism $M$ will always be clear from context.)
\end{definition}
\begin{definition}
Let $(n,\pi,k)$ be an instance and $M$ be a randomized direct revelation mechanism for that instance. Mechanism $M$ is \emph{dominant strategy incentive compatible (dominant strategy IC)} iff for all $i \in [n]$ and $\vec{s} \in \times_{j \in [n]} \text{supp}(\pi_j)$ and $\vec{v} \in \text{supp}(\pi)$,
\begin{equation*}
x_i(v_i,\vec{s}_{-i})v_i - p_i(v_i,\vec{s}_{-i}) \geq x_i(\vec{s})v_i - p_i(\vec{s}).
\end{equation*}
Mechanism $M$ is \emph{ex-post individually rational (ex-post IR)} iff for all $i \in [n]$ and $\vec{s} \in \text{supp}(\pi)$,
\begin{equation*}
x_i(s)v_i - p_i(s) \geq 0.
\end{equation*}
\end{definition}
For convenience we usually will not treat a mechanism as a probability distribution over outcomes, but rather as the result of a randomized procedure that interacts in some way with the buyers. In this case we say that a mechanism is \emph{implemented by} that procedure. The sequential posted price mechanisms are defined to be the mechanisms that are implemented by a particular such procedure, defined as follows.
\begin{definition}
An \emph{sequential posted price mechanism} for an instance $(n, \pi,k)$ is any mechanism that is implementable by iteratively selecting a buyer $i \in [n]$ that has not been selected in a previous iteration, and proposing a price $p_i$ for the service, which the buyer may accept or reject. If $i$ accepts, he gets the service and pays $p_i$, resulting in a utility of $v_i - p_i$ for $i$. If $i$ rejects, he pays nothing and does not get the service, resulting in a utility of $0$ for $i$. Once the number of buyers that have accepted an offer equals $k$, the process terminates. Randomization in the selection of the buyers and prices is allowed.  
\end{definition}
We will initially be concerned with only sequential posted price mechanisms.
Later in the paper we define and study the two generalizations of sequential posted price mechanisms that we mentioned in the introduction.

Note that each buyer in a sequential posted price mechanism has an obvious dominant strategy:
He will accept whenever the price offered to him does not exceed his valuation, and he will reject otherwise.
Also, a buyer always ends up with a non-negative utility when participating in a sequential posted price mechanism.
Thus, by the revelation principle (see, e.g., \citep{boergersmechdesign}),
a sequential posted price mechanism can be straightforwardly converted into a dominant strategy IC
and ex-post IR direct revelation mechanism that achieves the same expected revenue.

Our interest lies in analyzing the revenue performance of sequential posted price mechanisms.
We do this by comparing the expected revenue of such mechanisms to the maximum expected revenue that can be obtained
by a mechanism that satisfies dominant strategy IC and ex-post IR.
Thus, for a given instance let $OPT$ be the maximum expected revenue that can be attained by a dominant strategy IC,
ex-post IR mechanism and let $REV(\mathcal{C})$ be the maximum expected revenue achievable by some class of mechanisms $\mathcal{C}$.
Our goal throughout this paper is to derive instance-independent lower and upper bounds on the ratio $REV(\mathcal{C})/OPT$,
when $\mathcal{C}$ is the class of sequential posted price mechanisms or one of the generalizations mentioned.

A more general class of mechanisms is formed by the \emph{ex-post incentive compatible}, ex-post individually rational mechanisms.
\begin{definition}
Let $(n,\pi,k)$ be an instance and $M$ be a randomized direct revelation mechanism for that instance. Mechanism $M$ is \emph{ex-post incentive compatible (ex-post IC)} iff for all $i \in [n]$, $s_i \in \text{supp}(\pi_i)$ and $\vec{v} \in \text{supp}(\pi)$,
\begin{equation*}
x_i(\vec{v})v_i - p_i(\vec{v}) \geq x_i(s_i,\vec{v}_{-i})v_i - p_i(s_i,\vec{v}_{-i}).
\end{equation*}
\end{definition}
In other words, a mechanism is \emph{ex-post IC} if it is a pure equilibrium for the buyers to always report their valuation.
In this work we sometimes compare the expected revenue of our (dominant strategy IC and ex-post IR) mechanisms
to the maximum expected revenue of the more general class of ex-post IC, ex-post IR mechanisms.
This strengthens our positive results.
We refer the interested reader to \citet{roughgarden-talgamcohen} for a further discussion of and comparison between various solution concepts.

\section{Sequential posted price mechanisms}
We are interested in designing a posted price mechanism that, for any given $n$ and valuation distribution $\pi$, achieves an expected revenue that lies only a constant factor away from the optimal expected revenue that can be achieved by a dominant strategy IC,  ex-post IR mechanism. In this section we show that this is unfortunately impossible. In fact, 
we will show that the approximation ratios 
established in this section are asymptotically optimal in terms of a 
distributional parameter  
to be defined in section \ref{sec:sppmpositive}.

\subsection{Non-existence of good posted price mechanisms}\label{sec:nonexistence}
We next prove the following theorem.
\begin{theorem}
 For all $n \in \mathbb{N}_{\geq 2}$, there exists a valuation distribution $\pi$ such that for all $k \in [n]$ there does not exist a sequential posted price mechanism for instance $(n,\pi,k)$ that extracts a constant fraction of the expected revenue of the optimal dominant strategy IC, ex-post IR mechanism. 
\end{theorem}
\begin{proof}
We first consider the unit supply setting, i.e., instances of the form $(n,\pi,1)$. 
As a first step, we show that it is impossible to achieve a constant factor approximation when we compare a posted price mechanism to the expected expected \emph{optimal social welfare}, defined as:
$$
OSW = \mathbf{E}_{\vec{v} \sim \pi}[\max\{v_i \colon i \in [n]\}].
$$

Let $OR$ be the optimal revenue that a dominant strategy IC and ex-post IR mechanism can achieve. (Of course $OR$ depends on the valuation distribution $\pi$, but we assume that the valuation distribution is given, and implicit from context.)
It is clear that $OSW$ is an upper bound to $OR$ regardless of $\pi$, since a dominant strategy IC and ex-post IR mechanism will not charge (in expectation) any buyer a higher price than its expected valuation.

Fix $m \in \mathbb{N}_{\geq 1}$ arbitrarily, and consider the case where $n = 1$ and the valuation $v_1$ of the single buyer is taken from $\{1/a \colon a \in [m]\}$ distributed such that $\pi_1(1/a) = 1/m$ for all $a \in [m]$. In this setting, a posted price mechanism will offer the buyer a price $p$, which the buyer subsequently accepts iff $v_1 \geq p$. After that, the mechanism terminates.

Note that 
$
OSW = \frac{1}{m} \sum_{a = 1}^m \frac{1}{a}.
$
The expected revenue of the mechanism is
\begin{equation} \label{eq:rm}
RM = p\mathbf{Pr}_{v_1 \sim \pi_1}[v_1 \geq p] = p\frac{|\{a \colon 1/a \geq p\}|}{m} = \frac{|\{a \colon 1/a \geq 1/p^{-1}\}|}{mp^{-1}} = \frac{p^{-1}}{mp^{-1}} = \frac{1}{m}.
\end{equation}
Therefore:
$$
\lim_{m \rightarrow \infty} \frac{RM}{OSW} = \lim_{m \rightarrow \infty} \frac{1}{\sum_{a \in [m]} 1/a} = \frac{1}{H(m)} = 0.
$$
So, no posted price mechanism can secure in expectation a revenue that lies a constant factor away from the expected optimal social welfare. (Because our analysis is for an instance instance with only one buyer, this inapproximability result also holds for instances with independent valuations.)

We extend the above example in a simple way to a setting where the expected revenue of the optimal dominant strategy IC, ex-post IR mechanism is equal to the expected optimal social welfare.

Fix $m \in \mathbb{N}_{\geq 1}$ and consider a setting with $2$ buyers, where the type vector $(v_1,v_2)$ takes values in $\{(1/a,1/a) \colon a \in [m]\}$ according to the probability distribution where $\pi((1/a,1/a)) = 1/m$ for all $a \in [m]$. A mechanism that always gives buyer 1 the service, and charges buyer 1 the bid of buyer 2, is clearly dominant strategy IC and also clearly achieves a revenue equal to the optimal social welfare. 

In this two buyer setting, the value $OSW$ is again $OSW = \frac{1}{m} \sum_{a = 1}^m \frac{1}{a}$.
By symmetry, we may assume without loss of generality that a posted price mechanism works by first proposing a price $p_1$ to buyer $1$, and then proposing a price $p_2$ to buyer $2$ if buyer $1$ rejected the offer.
Using arguments similar to (\ref{eq:rm}), we derive that the revenue of this mechanism is:
\begin{align*}
RM & = p_1\mathbf{Pr}_{(v_1,v_2) \sim \pi}[v_1 \geq p_1] + p_2\mathbf{Pr}_{(v_1,v_2) \sim \pi}[v_1 < p_1 \cap v_2 \geq p_2] \\
 & = \frac{1}{m} + p_2\mathbf{Pr}_{(v_1,v_2) \sim \pi}[v_1 < p_1 \cap v_2 \geq p_2]
 \leq \frac{1}{m} + p_2\mathbf{Pr}_{(v_1,v_2) \sim \pi}[v_2 \geq p_2]
 = \frac{2}{m}.
\end{align*}
Therefore:
$$
\lim_{m \rightarrow \infty} \frac{RM}{OR} = \lim_{m \rightarrow \infty}\frac{RM}{OSW} \leq \lim_{m \rightarrow \infty} \frac{2}{\sum_{a \in [m]} 1/a} = 0.
$$

The above example establishes the non-existence of a good sequential posted price mechanism in the case where the service has to be provided to a single buyer.
Suppose now that the service can be provided to $2$ buyers, and each buyer gets the service at most once.
Consider again two buyers whose values are drawn from the probability distribution $\pi$ as defined above.
As above, by symmetry we may assume that our posted price mechanism first proposes price $p_1$ to buyer $1$,
and then proposes either price $p_2$ or $p_2'$ to buyer $2$: $p_2$ is proposed in case the offer was rejected by buyer $1$, and $p_2'$ is proposed otherwise. The difference with the previous analysis for the unit supply case is that the mechanism proposes a price to buyer $2$ regardless of whether buyer $1$ accepted the offer or not.

We derive:
\begin{align*}
RM & = p_1\mathbf{Pr}_{(v_1,v_2) \sim \pi}[v_1 \geq p_1 \cap v_2 < p_2'] + p_2\mathbf{Pr}_{(v_1,v_2) \sim \pi}[v_1 < p_1 \cap v_2 \geq p_2] \\
& \qquad  + (p_1 + p_2')\mathbf{Pr}_{(v_1,v_2) \sim \pi}[v_1 \geq p_1 \cap v_2 \geq p_2'] \\
 & \leq \frac{2}{m} + (p_1 + p_2')\mathbf{Pr}_{(v_1,v_2) \sim \pi}[v_1 \geq p_1 \cap v_2 \geq p_2'] \\
 & \leq \frac{2}{m} + 2\max\{p_1,p_2'\}\mathbf{Pr}_{(v_1,v_2) \sim \pi}[v_1 \geq \max\{p_1,p_2'\}] \leq \frac{4}{m}. 
\end{align*}
The optimal incentive compatible mechanism works by giving the service to both buyers while charging the bid of buyer 1 to buyer 2, and charging the bid of buyer 2 to buyer 1. The resulting expected revenue is exactly the expected optimal social welfare:
$
OR = OSW = \frac{1}{m} \sum_{a = 1}^m \frac{2}{a}.
$
We therefore obtain
$$
\lim_{m \rightarrow \infty} \frac{RM}{OR} = \lim_{m \rightarrow \infty}\frac{RM}{OSW} \leq \lim_{m \rightarrow \infty} \frac{4}{\sum_{a \in [m]} 2/a} = 0.
$$
The above yields an impossibility result for $2$-limited supply. By adding to this instance dummy buyers that always have valuation $0$, we obtain an impossibility result for $k$-limited supply, where $k \in \mathbb{N}$.
\end{proof}

We prove that the above impossibility result holds also in the continuous case even if we assume that all of the following conditions simultaneously hold: the valuation distribution is supported everywhere, is entirely symmetric, satisfies \emph{regularity}, satisfies the \emph{monotone hazard rate} condition, satisfies \emph{affiliation}, all the induced marginal distributions have finite expectation, and all the conditional marginal distributions are non-zero everywhere. We remark that \citet{roughgarden-talgamcohen} showed that when all these assumptions are simultaneously satisfied, the optimal ex-post IC and ex-post IR mechanism is the Myerson mechanism, that is the same mechanism that is optimal in the independent value setting. Thus, these conditions make the correlated setting very similar to the independent one with respect to revenue maximization.
Yet our result show that, whereas posted price mechanism can achieve a constant approximation revenue in the latter setting,
this result does not extend to the former one.

\begin{theorem}
\label{prop:spp_continuous}
There exists a valuation distribution $\pi$ such that
\begin{enumerate}
 \item $\pi$ has support $\left[0, 1\right]^n$;
 \item the expectation $\mathbf{E}_{\vec{v} \sim \pi}[v_i]$ is finite for any $i \in [n]$;
 \item $\pi$ is symmetric in all its arguments;
 \item $\pi$ is continuous and nowhere zero on $[0,1]^n$;
 \item the conditional marginal densities $\pi_{i \mid \vec{v}_{-i}}$ are nowhere zero for any $\vec{v}_{-i} \in [0,1]^{n-1}$ and any $i \in [n]$;
 \item $\pi$ has a monotone hazard rate and is regular;
 \item $\pi$ satisfies affiliation.
\end{enumerate}
for which there does not exist a sequential posted price mechanism in which valuations are distributed according to $\pi$ that extracts more than a constant fraction of the expected revenue of the optimal dominant strategy IC, ex-post IR mechanism.
\end{theorem}
\begin{proof}
Consider $c > 1$ and $m \geq (c - \log c)/2$ and set $M = 1 + 1/m$.
Let $V$ be a random variable whose value is drawn over the support $[1/m,1]$
according to the probability density function
$$
 f_V(x) = \frac{1}{(m-1)x^2}.
$$
Let $N_1$ and $N_2$ be two random variables whose values are independently drawn over the support $[0, M - v]$
according to the conditional density function
$$
 f_{N\mid V}(z \mid V = v) = \frac{c^{-z} \ln(c)}{Z(v)}
$$
with $c > 1$ and $Z(v) = 1 - c^{-(M-v)}$.
Finally, let $f$ be the probability density function of the pair $(X, Y) = (V+N_1-1/m, V+N_2-1/m)$.

Properties 1, 2 and 3 are trivial and can be immediately checked.
 
For $v \in [0,1]$, let $f_{X|V}(\cdot | V = v)$ and $f_{Y|V}(\cdot | V = v)$ be respectively the probability density functions of $X$ and $Y$ conditioned on the event that $V = v$. In order to establish the remaining properties, observe that
$f_{X \mid V}(x \mid V = v) = f_{N\mid V}(x - v + 1/m \mid V = v)$ if $x + 1/m \geq v$ and 0 otherwise.
Equivalently, $f_{Y \mid V}(y \mid V = v) = f_{N\mid V}(y - v + 1/m \mid V = v)$ if $y + 1/m \geq v$ and 0 otherwise.
Consider now the triple $(X, Y, V)$. The joint density function of this triple is
$$
 f_{X,Y,V}(x,y,v) = f_{X \mid Y,V}(x \mid Y = y, V = v) \cdot f_{Y\mid V}(y \mid V = v) \cdot f_V(v).
$$
Note that $f_{X \mid Y,V}(x \mid Y = y, V = v) = f_{X \mid V}(x \mid V = v)$ if $\min\{x,y\} + 1/m \geq v$ and 0 otherwise.
Then
$$
 f_{X,Y,V}(x,y,v) = f_{N\mid V}(x - v + 1/m\mid V = v) \cdot f_{N\mid V}(y - v + 1/m \mid V = v) \cdot f_{V}(v).
$$
if $\min\{x,y\} + 1/m \geq v$ and 0 otherwise. Hence, we can compute $f$ as follows:
$$
 f(x,y) = \int_{1/m}^1 f_{X,Y,V}(x,y,v) dv = \frac{\ln^2(c)}{m-1} \cdot c^{-(x+y)} \cdot \int_{1/m}^{\alpha+1/m} \frac{c^{2v}}{v^2 Z(v)^2} dv,
$$
where $\alpha = \min\left\{1-1/m, x, y\right\}$.
Note that the integrated function is continuous and positive in the interval in which it is integrated.
Hence, the integral turns out to be non-zero.
From this, we observe that $f(x,y)$ is continuous and nowhere zero on $[0,1]^2$, satisfying Property 4.

Let us now derive the conditional probability density functions. By symmetry it will be sufficient to focus only on $f_{X \mid Y}$.
\begin{align*}
 f_{X \mid Y}(x \mid Y = y) & = \int_{1/m}^{1} f_{X \mid Y,V}(x \mid Y = y, V = v) \cdot f_V(v) dv = \frac{\ln(c)}{m-1} \cdot c^{-x} \cdot \int_{1/m}^{\alpha+1/m} \frac{c^v}{v^2 Z(v)} dv\\
 & = \frac{m^2 c^M \ln(c)}{m-1} \cdot c^{-x} \cdot \int_{0}^{\alpha} \frac{1}{(mz+1)^2 (c^{1-z} - 1)} dz.
\end{align*}
It is now obvious that the conditional probability density functions are continuous and nowhere zero, as desired by Property 5.
 
Let $\gamma(z) = 1/((mz+1)^2 (c^{1-z} - 1))$, $g(a) = \int_0^a \gamma(z) dz$ with $a \in [0,1]$ and let $\alpha' = \min\{y, 1-1/m\}$.
Then
$$
 f_{X \mid Y}(x \mid Y = y) = \frac{m^2 c^M \ln(c)}{m-1} \cdot c^{-x} \cdot \begin{cases}
                                                                     g(x), & \text{if } x < \alpha';\\
                                                                     g(\alpha'), & \text{otherwise}.
                                                                    \end{cases}
$$
Moreover, we have that
\begin{align*}
 & 1 - F_{X \mid Y}(x \mid Y = y) = \int_x^1 f_{X \mid Y}(z \mid Y = y) dz \\
 & \qquad = \frac{m^2 c^M \ln(c)}{m-1} \cdot \begin{cases}
                                                                                                     \int_x^{\alpha'} c^{-z} g(z) dz + g(\alpha') \int_{\alpha'}^1 c^{-z} dz, & \text{if } x < \alpha';\\
                                                                                                     g(\alpha') \int_x^1 c^{-z} dz = \frac{g(\alpha') (c^{-x} - c^{-1})}{\ln(c)}, & \text{otherwise}.
                                                                                                    \end{cases}
\end{align*}
Hence, the inverse hazard rate is
$$
 I(x) = \frac{1 - F_{X \mid Y}(x \mid Y = y)}{f_{X \mid Y}(x \mid Y = y)} = \begin{cases}
                                                                              \frac{\int_x^{\alpha'} c^{-z} g(z) dz + g(\alpha') \int_{\alpha'}^1 c^{-z} dz}{c^{-x} g(x)}, & \text{if } x < \alpha';\\
                                                                              \frac{1 - c^{x-1}}{\ln(c)}, & \text{otherwise}.
                                                                             \end{cases}
$$
We prove that $I(x)$ is non-increasing in $x$ in the interval $[0,1]$ and thus $f$ has the monotone hazard rate and is, as a consequence, regular, as required by Property 6.
 
Clearly, $I(x)$ is non-increasing in $x$ in the interval $[\alpha', 1]$ since in this case $I(x) = (1 - c^{x-1})/\ln(c)$.
Moreover, $I(x)$ does not have discontinuities for $x = \alpha'$.
So, it is sufficient to show that $I(x)$ in non-increasing also in the interval $[0, \alpha']$.
To this aim, observe that for $x < \alpha'$,
\begin{align*}
 \frac{d I(x)}{dx} & = \frac{d}{dx} \frac{\int_x^{\alpha'} c^{-z} g(z) dz + g(\alpha') \int_{\alpha'}^1 c^{-z} dz}{c^{-x} g(x)}\\
 & = \left(c^{-x} g(x) \frac{d}{dx} \int_x^{\alpha'} c^{-z} g(z) dz - \int_x^{\alpha'} c^{-z} g(z) dz \frac{d}{dx} c^{-x} g(x) \right. \\
 & \qquad \left. + c^{-x} g(x) g(\alpha') \frac{d}{dx} \int_{\alpha'}^1 c^{-z} dz - g(\alpha') \int_{\alpha'}^1 c^{-z} dz \frac{d}{dx} c^{-x} g(x)\right)/(c^{-x} g(x))^2.
\end{align*}
Observe that, according to the second fundamental theorem of calculus,
$$
 \frac{d}{dx} \int_x^{\alpha'} c^{-z} g(z) dz = - c^{-x} g(x),
$$
whereas
$$
 \frac{d}{dx} c^{-x} g(x) = c^{-x} \gamma(x) - c^{-x} g(x) \ln(c),
$$
and $\frac{d}{dx} \int_{\alpha'}^1 c^{-z} dz = 0$.
Then,
$$
 \frac{d I(x)}{dx} = -1 + \left(\ln(c) - \frac{\gamma(x)}{g(x)}\right) \frac{\int_x^{\alpha'} c^{-z} g(z) dz + g(\alpha') \int_{\alpha'}^1 c^{-z} dz}{c^{-x} g(x)} = -1 + \left(\ln(c) - \frac{\gamma(x)}{g(x)}\right) I(x).
$$
The result then follows by showing that $\gamma(x)/g(x) \geq \ln(c)$.
 
To this aim, let us consider the function $\gamma'(z) = c^{z}(c^{1-z}-1) (mz+1)^2$ for $z \in [0,x]$.
Note that
$$
 \frac{d \gamma'(z)}{dz} = 2mc(mz+1) \left(1 - \frac{2m + mz \log c + \log c}{c^{1-z}}\right) \leq 2mc(mz+1) \left(1 - \frac{2m + \log c}{c}\right) \leq 0,
$$
where we used the fact that $m \geq (c - \log c)/2$. Thus, $\gamma'(z)$ is non-increasing in its argument and, in particular,
$$
 \gamma'(z) \geq \gamma'(x) \geq (c^{x} - 1) (c^{1-x}-1) (mx+1)^2.
$$
By simple algebraic manipulation, it then follows that $\gamma(z) \leq \gamma(x) c^z/(c^{x} - 1) $.
Then
$$
 \frac{\gamma(x)}{g(x)} = \frac{\gamma(x)}{\int_0^x \gamma(z) dz} \geq \frac{\gamma(x)}{\int_0^x \frac{c^z}{c^{x} - 1} \gamma(x) dz} = \frac{c^{x} - 1}{\int_0^x c^z dz} = \ln(c),
$$
as desired.
 
Set now $C=\ln^2(c)/(m-1)$ and
let $h(a) = \int_{1/m}^{\min\{1-1/m, a\} + 1/m} \frac{c^{2v}}{v^2 Z(v)^2} dv$ with $a \in [0,1]$.
Note that integrated function is positive for any $v \in [1/m, 1]$.
Hence, the integral increases as the size of the interval in which it is defined increases.
In other word, the function $h(a)$ is non-decreasing in $a$.

Consider now the two pairs $(x,y)$ and $(x',y')$.
Moreover, let $\hat{x} = \max\{x, x'\}$ and $\check{x} = \min\{x, x'\}$ and, similarly, define $\hat{y}$ and $\check{y}$.
We show that $f(x,y) f(x',y') \leq f(\hat{x},\hat{y}) f(\check{x},\check{y})$,
satisfying in this way also Property 7.
 
Indeed,
$$
 f(x,y) f(x',y') = C^2 c^{-(x+y+x'+y')} h(\min\{x,y\})h(\min\{x',y'\}).
$$
If $x \geq x'$ and $y \geq y'$ ($x < x'$ and $y < y'$, respectively),
then $(\hat{x}, \hat{y}) = (x,y)$ ($(x',y')$, resp.) and $(\check{x}, \check{y}) = (x',y')$ ($(x,y)$, resp.),
and the desired result immediately follows.
Suppose instead that $(\hat{x}, \hat{y}) = (x, y')$ and $(\check{x}, \check{y}) = (x', y)$.
Then
$$
 f(\hat{x}, \hat{y}) f(\check{x}, \check{y}) = C^2 c^{-(x+y+x'+y')} h(\min\{x,y'\})h(\min\{x',y\}).
$$
We will prove that in this case $h(\min\{x,y\})h(\min\{x',y'\}) \leq h(\min\{x,y'\})h(\min\{x',y\})$.
First observe that on both sides one of the two factors must be $h(\min\{x,y,x',y'\}) = h(\min\{x',y\})$.
Suppose without loss of generality, that $\min\{x',y\} = y$.
Then it is sufficient to prove that $h(\min\{x',y'\}) \leq h(\min\{x,y'\})$,
or, since $h$ in non-decreasing, that $\min\{x',y'\} \leq \min\{x, y'\}$.
If $x \leq y'$, then $x' \leq x \leq y'$ by hypothesis and the claim follows.
If $y' < x$, then it immediately follows that $\min\{x',y'\} \leq y'$.
The case that $(\hat{x}, \hat{y}) = (x', y)$ and $(\check{x}, \check{y}) = (x, y')$ is similar and hence omitted.
 
Finally, observe that $\lim_{c \rightarrow \infty} f_{N \mid V}(0 \mid V = v) = 1$
and $\lim_{c \rightarrow \infty} f_{N \mid V}(z \mid V = v) = 0$ for any $z > 0$.
Hence, $\lim_{c \rightarrow \infty} X = \lim_{c \rightarrow \infty} Y = V - 1/m$.

Let us consider the case that the service can be offered to only one buyer.
In this setting, the following is a dominant strategy IC and ex-post IR mechanism: it offers
to buyer 1 the service at a price of the valuation of buyer 2 minus a fixed constant $\epsilon$. For small enough $c$, $\epsilon$ can be chosen arbitrarily small. Thus, for any $\epsilon$ there exists a choice of $c$ such that in expectation this mechanism extracts as revenue all but $\epsilon$ of the social welfare.
The expected optimal social welfare (and thus the optimal expected revenue) is:
$$
OR = OSW = \mathbf{E}_{v \sim f_V}[v - 1/m] = \int_{1/m}^1 \frac{v - 1/m}{(m-1)v^2} dv = \frac{\ln(m)}{m-1} - \frac{1}{m} \geq \frac{\ln(m - 1)}{m-1}.
$$
A posted price mechanism will offer buyer 1 a price $p_1\geq 0$, which the buyer subsequently accepts iff $X \geq p_1$. After that, if buyer 1 rejects, the mechanism offers a price $p_2$ to buyer 2. Thus, if $p_1 \in [0, 1-1/m]$, then
\begin{align*}
p_1\mathbf{Pr}_{v \sim f_V}[X \geq p_1] & = p_1 \mathbf{Pr}_{v \sim f_V}[v \geq p_1 + 1/m] \\
& = p_1 \int_{p_1+1/m}^{1} \frac{1}{(m-1) v^2} dv = \frac{p_1}{m-1} \left(\frac{1}{p_1 + 1/m} - 1\right) \leq \frac{1-p_1}{m-1} \leq \frac{1}{m-1}.
\end{align*}
Moreover, if $p_1 \geq 1 - 1/m$, then $p_1\mathbf{Pr}_{v \sim f}[X \geq p_1] = 0$. Hence,
$$
RM = p_1\mathbf{Pr}_{v \sim f}[v \geq p_1] + p_2\mathbf{Pr}_{v \sim f}[v < p_1 \cap v \geq p_2] \leq \frac{1}{m-1} + p_2\mathbf{Pr}_{v \sim f}[v \geq p_2] \leq  \frac{2}{m-1}.
$$
Therefore:
\[
\lim_{m \rightarrow \infty} \frac{RM}{OR} = \lim_{m \rightarrow \infty}\frac{RM}{OSW} \leq \lim_{m \rightarrow \infty} \frac{2}{\ln(m - 1)} = 0.
\]

The case in which it can be offered to both buyer is similar and omitted.
\end{proof}

We note that above theorem can be easily extended to a any number of buyers,
by adding dummy buyers whose valuation is independently drawn over the support $[0,1]$
according to the probability density function
$$
 f(x) = \frac{c^{-x} \ln c}{1 - c^{-1}}.
$$
Note that, with this extension, the resulting distribution $\pi$ does not satisfy the symmetry condition anymore. In the result of \citet{roughgarden-talgamcohen}, symmetry is not necessary for the optimality of the Myerson mechanism to hold.

\subsection{A revenue guarantee for sequential posted price mechanisms}
\label{sec:sppmpositive}
In the previous section we have demonstrated that it is impossible to have
a sequential posted price mechanism extract a constant fraction of the optimal revenue.
More precisely, in our example instances it was the case that the expected revenue
extracted by every posted price mechanism is a $\Theta(1/\log(r))$ fraction of the optimal expected revenue,
where $r$ is the ratio between the highest valuation and the lowest valuation in the support of the valuation distribution.
A natural question that arises is whether this is the worst possible instance in terms of revenue extracted, as a function of $r$.
We show here that this is indeed the case, asymptotically:
For every valuation distribution $\pi$, there exists a mechanism that extracts in expectation
at least a $\Theta(1/\log(r))$ fraction of the revenue of the optimal revenue.
We note that in many realistic scenarios, we do not expect the extremal valuations of the buyers to lie too far from each other,
because often the valuation of a buyer is strongly impacted by prior objective knowledge of the value of the service to be auctioned.
The results of this section are valuable when that is the case. 

We start with the unit supply case.
\begin{definition}
For a valuation distribution $\pi$ on $\mathbb{R}^n$, let $v_{\pi}^{\max}$ and $v_{\pi}^{\min}$ be $\max\{v_i \colon v \in \text{supp}(\pi), i \in [n]\}$ and $\min\{\max\{v_i \colon i \in [n]\} \colon v \in \text{supp}(\pi)\}$ respectively. Let $r_{\pi} = v_{\pi}^{\max}/v_{\pi}^{\min}$ be the ratio between the highest and lowest coordinate-wise maximum valuation in the support of $\pi$.
\end{definition}

\begin{proposition}\label{ubsppunitsupply}
Let $n \in \mathbb{N}_{\geq 1}$, and let $\pi$ be a probability distribution on $\mathbb{R}^n$.
For the unit supply case there exists a sequential posted price mechanism that, when run on instance $(n,\pi,1)$, extracts in expectation at least an $\Omega(1/\log(r_{\pi}))$ fraction of the expected revenue of the expected optimal social welfare (and therefore also of the expected revenue of the optimal dominant strategy IC and ex-post IR auction).
\end{proposition}
\begin{proof}
The proof use a standard bucketing trick (see, e.g., \cite{secretary}).
Specifically, let $M$ be the sequential posted price mechanism that draws a value $p$ uniformly at random from the set $S = \{v_{\pi}^{\min}2^k \colon k \in [\lceil \log(r_{\pi}) - 1 \rceil] \cup \{0\}\}$. $M$ offers price $p$ to all the bidders in an arbitrary order, until a bidder accepts.

Let $\pi_{\max}$ be the probability distribution of the coordinate-wise maximum of $\pi$.
Note that $|S|$ does not exceed $\log(r_{\pi}^{\max})$. Therefore the probability that $p$ is the highest possible value (among the values in $S$) that does not exceed the value drawn from $\pi_{\max}$, is equal to $1/\log(r_{\pi})$. More formally, let $\pi_S$ be the probability distribution from which $p$ is drawn; then
\begin{equation*}
\mathbf{Pr}_{v^{\max} \sim \pi_{\max}, p \sim \pi_S }[p \leq v^{\max} \cap (\not\exists p' \in S \colon p' > p \wedge p' \leq v^{\max})] \leq \frac{1}{\log(r_{\pi})}.
\end{equation*}
Thus, with probability $1/\log(r_{\pi})$, the mechanism generates a revenue of exactly $v_{\pi}^{\min}2^k$, where $k$ is the number such that the value drawn from $\pi_{\max}$ lies in between $v_{\pi}^{\min}2^k$ and $v_{\pi}^{\min}2^{k+1}$. This implies that with probability $1/\log(r_{\pi})$ the mechanism generates a revenue that lies a factor of at most $1/2$ away from the optimal social welfare $OPT(\vec{v})$ (i.e., the coordinate-wise maximum valuation):
\begin{equation*}
\mathbf{E}_{\vec{v} \sim \pi}[\text{revenue of }M(\vec{v})] \geq \frac{1}{\log(r_{\pi})}\frac{1}{2}\mathbf{E}_{\vec{v} \sim \pi}[OPT(\vec{v})] \geq \frac{1}{2\log(r_{\pi})}\mathbf{E}_{\vec{v} \sim \pi}[OPT(\vec{v})].\qedhere
\end{equation*}
\let\qed\relax
\end{proof}
This result can be generalized to yield revenue bounds for the case of $k$-limited supply, where $k > 1$.
We prove a more general variant of the above in Appendix~\ref{apx:bound_limited}.

The above result does not always guarantee a good revenue; for example in the extreme case where $v_{\pi}^{\min} = 0$.
However, it is straightforward to generalize the above theorem such that it becomes useful for a much bigger family of probability distributions:
let $\hat{v}$ and $\check{v}$ be two any two values in the support of $\pi_{\max}$,
and let $c(\hat{v}, \check{v}) = \mathbf{Pr}_{v^{\max} \sim \pi_{\max}}[\check{v} \leq v^{\max} \leq \hat{v}]$.
Then by replacing the values $v_{\pi}^{\max}$ and $v_{\pi}^{\min}$ in the above proof by respectively $\hat{v}$ and $\check{v}$,
we obtain a sequential posted price mechanism that extracts in expectation
a $c(\hat{v},\check{v})/(2\log( \hat{v}/ \check{v}))$ fraction of the optimal social welfare.
By choosing $\hat{v}$ and $\check{v}$ such that this ratio is maximized,
we obtain a mechanism that extracts a significant fraction of the optimal social welfare
in any setting where the valuation distribution of a buyer is concentrated in a relatively not too large interval.

A better result can be given for the unlimited supply case.
\begin{definition}
For a valuation distribution $\pi$ on $\mathbb{R}^n$ and any $i \in [n]$, let $v_{\pi,i}^{\max}$ and $v_{\pi,i}^{\min}$ be $\max\{v_i \colon \vec{v} \in \text{supp}(\pi)\}$ and $\min\{v_i \colon \vec{v} \in \text{supp}(\pi)\}$ respectively. Let $r_{\pi,i} = v_{\pi,i}^{\max}/v_{\pi,i}^{\min},$ be the ratio between the highest and lowest valuation of buyer $i$ in the support of $\pi$.
\end{definition}

\begin{proposition}\label{prop:logunlimited}
Let $n \in \mathbb{N}_{\geq 1}$, and let $\pi$ be a probability distribution on $\mathbb{R}^n$. 
There exists a sequential posted price mechanism that, when run on instance $(n, \pi, n)$, extracts in expectation at least an $\Omega(1/\log(\max\{r_{\pi,i} \colon i \in [n]\}))$ fraction of the expected revenue of the expected optimal social welfare (and therefore also the expected revenue of the optimal dominant strategy IC and ex-post IR mechanism).
\end{proposition}

The proof (found in Appendix~\ref{sec:logunlimited}) applies similar techniques as is done in the proof of Proposition \ref{ubsppunitsupply}.
The techniques used for proving these results can be applied to improve the approximation guarantees for the more general $k$-limited supply setting, for any $k \in [n]$, under special conditions. We give an example of such a result in Appendix~\ref{apx:bound_limited}.

Clearly, the stated bound of $O(1/\log(\max\{r_{\pi,i} \colon i \in [n]\}))$ is very crude.
For most practical settings we expect that it is possible to do a much sharper revenue analysis of the mechanisms in the proofs of the above propositions, by taking the particular valuation distribution into account.
Moreover, as suggested above, also for the unlimited supply case it is possible to tweak the mechanism 
in a straightforward way in order to achieve a good revenue in cases where the ratios $r_{\pi,i}$ are very large.
Finally, note that the mechanisms in the proofs of these two propositions do not take into account any dependence and correlation among the valuations of the buyers.
When provided with a particular valuation distribution,
a better revenue and sharper analysis may be obtained by taking such dependence into account,
and adapting the mechanisms accordingly.

\section{Enhanced sequential posted price mechanisms}\label{sec:generalization}
The negative results on sequential posted price mechanisms suggest that
it is necessary for a mechanism to have a means to retrieve the valuations of some of the buyers in order to improve the revenue performance.
We accordingly propose a generalization of sequential posted price mechanisms,
in such a way that they possess the ability to retrieve valuations of some buyers.
\begin{definition}
An \emph{enhanced sequential posted price mechanism} for an instance $(n, \pi,k)$ is a mechanism that can be implemented by iteratively selecting a buyer $i \in [n]$ that has not been selected in a previous iteration, and performing exactly one of the following actions on buyer $i$:
\begin{itemize}
\item Propose service at price $p_i$ to buyer $i$, which the buyer may accept or reject. If $i$ accepts, he gets the service and pays $p_i$, resulting in a utility of $v_i - p_i$ for $i$. If $i$ rejects, he pays nothing and does not get the service, resulting in a utility of $0$ for $i$.
\item Ask $i$ for his valuation (in which case the buyer pays nothing and does not get service).
\end{itemize}
Randomization is allowed.  
\end{definition}
This generalization is still dominant strategy IC and ex-post IR
(that is, the revelation principle allows us to convert them to dominant strategy IC, ex-post IR direct revelation mechanisms).
Indeed, for enhanced sequential posted price mechanisms, when a buyer gets asked his valuation,
he has no incentive to lie, because in this case he does not get service and he pays nothing\footnote{A problematic aspect is that while there is no incentive to a buyer to lie, there is also no incentive to tell the truth.
This problem is addressed in Appendix~\ref{apx:espp_fix}}.

Next we analyze the revenue performance of enhanced sequential posted price mechanisms.
For this class of mechanisms we prove that, it is unfortunately still the case that no constant fraction of the optimal revenue can be extracted.
Specifically, the next section establishes an $O(1/n)$ bound for enhanced sequential posted price mechanisms.
However, enhanced sequential posted price mechanisms turn out to be more powerful than the standard sequential posted price.
Indeed, contrary to what we had for the former ones,
the enhanced mechanisms can be shown to extract a fraction of the optimal revenue that is independent of the valuation distribution.
More precisely, the $O(1/n)$ bound turns out to be asymptotically tight.
Our main positive result for enhanced sequential posted price mechanisms is that when dependence of the valuation among the buyers is limited,
then a constant fraction of the optimal revenue can be extracted.
Specifically, in Section \ref{sec:generalizationpositive} we define the concept of $d$-dimensional dependence and  prove that for an instance $(n,\pi,k)$ where $\pi$ is $d$-dimensionally dependent,  
there exists an enhanced sequential posted price mechanism that extracts an $\Omega(1/d)$ fraction of the optimal revenue.
(This result implies the claimed $\Omega(1/n)$ bound by taking $d = (n-1)$.)

It is natural to identify the basic reason(s) why, in the case of
general correlated distributions, standard and enhanced sequential posted price
mechanisms fail to achieve a constant approximation of the optimum revenue. 
There are two main limitations of (enhanced) sequential 
posted price mechanisms; namely, 
1) such mechanisms do not solicit bids or values from all buyers, and
2) such mechanisms award items in a sequential
manner. Although it is crucial to retrieve the valuation of \emph{all} (but one of the) buyers, 
we now show, in contrast to previously known approximation results,  
that it is possible to achieve a constant fraction of the 
optimum revenue by a mechanism that allocates
items sequentially, and moreover, in an on-line manner. 
Specifically, we consider the following superclass of the enhanced sequential posted price mechanisms.
\begin{definition}
Let $(n,\pi,k)$ be an instance and let $M$ be a mechanism for that instance.
Mechanism $M$ is a \emph{blind offer mechanism} iff it can be implemented as follows. 
Let $\vec{b}$ be the submitted bid vector:
\begin{enumerate}
\item Terminate if $\vec{b} \not\in \text{supp}(\pi)$.
\item Either terminate or select a buyer $i$ from the set of buyers that have not yet been selected, such that the choice of $i$ does not depend on $\vec{b}$.
\item Propose buyer $i$ to offer service at price $p_i$, where $p_i$ is drawn from a probability distribution that depends only on $\pi_{i,\vec{b}_{-i}}$ (hence the distribution where $p_i$ is drawn from is determined by $\vec{b}_{-i}$ and in particular it does not depend on $b_i$).
\item Go to step $2$ if there is supply left, i.e., if the number of buyers who have accepted offers does not exceed $k$.
\end{enumerate}
Randomization is allowed.
\end{definition}
\begin{remark}\label{rem:blindofferic}
The price offered to a buyer is entirely determined by the valuations of the remaining buyers,
and is independent of what is reported by buyer $i$ himself. Also the iteration in which a buyer is picked cannot be influenced by his bid.
Nonetheless, blind offer mechanisms are in general not incentive compatible due to the fact that a bidder may be incentivized to misreport his bid in order to increase the probability of supply not running out before he is picked. However, blind offer mechanisms can easily be made incentive compatible as follows: let $M$ be a non-IC blind offer mechanism, let $\vec{b}$ be a bid vector and let $z_{i}(\vec{b})$ be the probability that $M$ picks bidder $i$ before supply has run out. When a bidder is picked, we adapt $M$ by \emph{skipping} that bidder with a probability $p_i(\vec{b})$ that is chosen in a way such that $z_i(\vec{b})p_i(\vec{b}) = \min\{z_i(b_i\vec{b}_{-i}) \colon b_i \in \text{supp}(\pi_i)\}$. This is a blind offer mechanism in which buyer $i$ has no incentive to lie, because now the probability that $i$ is made an offer is independent of his bid. Doing this iteratively for all buyers yields a dominant strategy IC mechanism $M'$. Note that the act of \emph{skipping} a bidder can be implemented by offering a price that is so high that a bidder will never accept it, thus $M'$ is still a blind offer mechanism. Moreover, if the probability that any bidder in $M$ is made an offer is lower bounded by a constant $c$, then in $M'$ the probability that any bidder is offered a price is at least $c$. We apply this principle in the proof of Theorem \ref{mainprop-blindoffer} below in order to obtain a dominant strategy IC mechanism with a constant factor revenue performance\footnote{More precisely, for the particular (non-truthful) blind offer mechanism that we propose and analyze in section \ref{sec:blind}, it will turn out that applying the transformation described here does not result in any additional loss in revenue.}.
\end{remark}

It is well known under the name of the \emph{taxation principle} that
a mechanism is dominant strategy IC if and only if
it can be implemented by an algorithm that works as follows:
(i) the buyers are \emph{simultaneously} presented with a payment that does not depend on their own bid;
(ii) the items are allocated to the buyers for which the profit, i.e., the difference between the bid and the price, is maximized.
This algorithm closely resembles the description of blind offer mechanisms.
However, we would like to emphasize that there are some significant differences.
First, using the taxation principle, prices are set in advance of offering 
service,
whereas in blind offer mechanisms prices are presented sequentially,
thus the price offered to the $i$-th buyer can depend on the decisions
taken by the previous buyers.
Second, using the taxation principle the winners are chosen at the end,
whereas in blind offer mechanisms the winners are decided on-line,
so that it is possible that items are allocated to buyers without maximizing the profit.
Hence,  there are dominant strategy mechanisms
cannot be implemented as blind offer mechanisms.

A blind offer mechanisms preserves the same conceptually simple structure as standard and enhanced posted price mechanisms,
but it bases its proposal to a buyer $i$ on the set of \emph{all} valuations other than that of $i$,
 whereas an ESPP mechanism bases its proposal
to a buyer $i$ only on the valuations that have been \emph{revealed} by buyers in \emph{previous} iterations.
This increases the power of blind offer mechanisms: for example,
it is not hard to see that the classical Myerson mechanism for the \emph{independent} single-item setting belongs to the class of blind offer mechanisms.
Thus blind offer mechanisms are optimal when buyers' valuations are independent.
We will prove next that, even when buyer valuations are \emph{correlated},
blind offer mechanisms can always extract a constant fraction of the optimal revenue,
even for the ex-post IC, ex-post IR solution concept. 
For
correlated valuation distributions, other mechanisms that achieve a constant approximation
to the optimal revenue have been defined by \cite{ronen}, and then by \cite{chawla2014approximate} and \cite{dobzinski2011}.
However, these mechanisms, as in the taxation principle setting, allocate the items to profit-maximizing buyers.
Thus, they are different from blind offer mechanisms in which the allocation is on-line.

\subsection{Limitations of enhanced sequential posted price mechanisms}\label{sec:generalizationnegative}
Here we show that enhanced sequential posted price mechanisms cannot extract a constant fraction
of the expected revenue of the optimal dominant strategy IC, ex-post IR mechanisms.
This is done by constructing a family of instances on which no enhanced sequential posted price mechanism can perform well.

\begin{theorem}\label{thm:enhancedlb}
For all $n \in \mathbb{N}_{\geq 2}$, there exists an valuation distribution $\pi$ such that for all $k \in [n]$ there does not exist a enhanced sequential posted price mechanism for instance $(n,\pi,k)$ that extracts more than a $O(1/n)$ fraction of the expected revenue of the optimal dominant strategy IC,
ex-post IR mechanism.
\end{theorem}
\begin{proof}
We prove this for the case of $k = n$. The proof is easy to adapt for different $k$.

Let $n \in \mathbb{N}$ and $m = 2^n$. We prove this claim by specifying an instance $I_n$ with $n$ buyers, and proving that $\lim_{n \rightarrow 0} RM(I_n)/OR(I_n) = 0$, where $RM(I_n)$ denotes the largest expected revenue achievable by any enhanced sequential posted price mechanism on $I_n$, and $OR(I_n)$ denotes the largest expected revenue achievable by a dominant strategy IC, ex-post IR mechanism. 

$I_n$ is defined as follows. Let $\epsilon \in \mathbb{R}_{> 0}$ be a number smaller than $1/nm^2$. The valuation distribution $\pi$ is the one induced by the following process: (i) Draw a buyer $i^\star$ from the set $[n]$ uniformly at random; (ii) Draw numbers $\{c_j \colon j \in [n]\setminus\{i^\star\}\}$ independently from $[m]$ uniformly at random; (iii) For all $j \in [n] \setminus \{i^\star\}$, set $v_j = c_j\epsilon$; (iv) Set 
\begin{equation}\label{modval}
v_{i^\star} = \frac{1}{\left(\sum_{j \in [n]\setminus\{i^\star\}} c_j \right)_{\text{mod } m} + 1}.
\end{equation}
Observe that for this distribution it holds that for all $i \in [n]$, the valuation $v_{i}$ is uniquely determined by the valuations  $(v_{j})_{j \in [n]\setminus \{i\}}$.
The optimal (direct revelation) mechanism can therefore extract the total optimal social welfare as its revenue, as follows:
it provides service to every buyer, and sets the payment as follows.
Let $b_i$ be the bid, i.e., the reported valuation, of buyer $i$. Then,
\begin{itemize}
\item if $b_{j} < 1/m$ for all $j \in [n]\setminus\{i\}$, charge $i$ a price of 
$
1/\left(\left(\sum_{j \in [n]\setminus\{i\}} b_j/\epsilon \right)_{\text{mod } m} + 1\right);
$
\item otherwise, if there is a buyer $j \in [n]\setminus\{i\}$ and a number $c_{i} \in [m]$ such that 
$
b_{j} = 1/\left(\left(c_i + \sum_{\ell \in [n]\setminus\{i,j\}} b_{\ell}/\epsilon\right)_{\text{mod } m} + 1\right),
$
then charge $i$ the price $c_{i}\epsilon$;
\item otherwise, the mechanism charges $i$ an arbitrary price.
\end{itemize}
This mechanism is dominant strategy IC because the mechanism's decision to provide service to a buyer does not depend on his bid, and the price that a buyer is charged is not dependent on his own bid. This mechanism is ex-post IR because bidding truthfully always gives the buyer a utility of $0$. This mechanism achieves a revenue equal to the optimal social welfare because (by definition of the pricing rule) the price that a buyer is charged is equal to the valuation of that buyer, if all buyers bid truthfully. Also, note that the third bullet in the above specification of the mechanism will not occur when the buyers bid truthfully, and is only included for the sake of completely specifying the mechanism.

We argue that $OR(I_n) = \mathbf{E}_{\vec{v} \sim \pi}\left[\sum_{i \in [n]} v_i\right] = \sum_{i \in [n]} \mathbf{E}_{\vec{v} \sim \pi}[v_i] = (n-1)m\epsilon/2 + H_m/m$, where the last equality follows because the expected valuation of each of the buyers is $((n-1)/n)(m\epsilon/2) + (1/n)(H_m/m)$. This in turn holds because a buyer is elected as buyer $i^\star$ with probability $1/n$, and buyer $i^\star$'s marginal distribution is the distribution $\pi'$ induced by drawing a value from the set $\{1/a \colon a \in [m]\}$ uniformly at random. The latter distribution has already been encountered in the beginning of Section \ref{sec:nonexistence}, where we concluded that its expected value is $H_m/m$.

We now proceed to prove an upper bound on $RM(I_n)$. Let $M$ be an arbitrary enhanced posted price mechanism. Because $M$ is randomized, running $M$ on $I_n$ can be viewed as a probability distribution on a sample space of deterministic enhanced posted price mechanisms that are run on $I_n$.
We analyze the revenue of the mechanism conditioned on three disjoint events that form a partition of this sample space. Consider first the event $E_1$ that buyer $i^\star$ gets asked for his valuation (when running $M$ on $I_n$). Conditioned on this event, the mechanism does not attain a revenue of more than $(n-1)m\epsilon$ because the revenue of each buyer in $[n] \setminus \{i^\star\}$ is at most $m\epsilon$. 

Consider next the event $E_2$ where buyer $i^\star$ does not get asked for his valuation and buyer $i^\star$ is not the last buyer that is selected. Then a price $p_{i^\star}$ is proposed to $i^\star$. Without loss of generality, $M$ draws $p_{i^\star}$ from a probability distribution $P_{i^\star}$ with finite support, and the choice of distribution $P_{i^\star}$ depends on the sequence $S$ of buyers queried prior to $i^\star$ together with the responses of the buyers in $S$. These responses take the form of a reported valuation in case a buyer in $S$ is asked to report his valuation, and the form of an accept/reject decision otherwise. Because $i^\star$ is not the last buyer selected, $[n]\setminus (S\cup\{i^\star\})$ is non-empty, and there exists a buyer $j \in [n]\setminus (S \cup \{i^\star\})$ such that the choice of $P_{i^\star}$ does not depend on $c_j$. By the fact that $c_j$ is drawn independently and uniformly at random from $[m]$ for all $j \in [n]\setminus (S \cup \{i^\star\})$ and by \eqref{modval}, the marginal probability distribution of the valuation of buyer $i^\star$ conditioned on $E_2$, is $\pi'$ (which we defined above). Therefore 
\begin{equation*}
\mathbf{E}_{p_{i^\star} \sim P_{i^\star}, \vec{v} \sim \pi}[p_{i^\star}\mathbf{Pr}[v_{i^\star} \leq p_{i^\star}]] = \mathbf{E}_{p_{i^\star} \sim P_{i^\star}, \vec{v} \sim \pi'}[p_{i^\star}\mathbf{Pr}[v_{i^\star} \leq p_{i^\star}]] = \frac{1}{m},
\end{equation*}
where the last equality follows from (\ref{eq:rm}). Thus, the expected revenue of $M$ conditioned on $E_2$ is at most $1/m + (n-1)m\epsilon$.

In the event $E_3$, the mechanism selects $i^\star$ last. The expected revenue of $M$ conditioned on this event is at most the expected maximum social welfare: $(n-1)m\epsilon/2 + H_m/m$. The probability of event $E_3$ occurring is $1/n$, because of the following. For $\ell \in [n]$, let $E_3^\ell$ be the event that $i^\star$ is not the $\ell$-th buyer selected by $M$, and let $E_3^{<\ell}$ be the event that $i^\star$ is not among the first $\ell-1$ buyers selected by $M$. Note that this means that $\mathbf{Pr}[E_3^{<1}] = 1$. Then,
$$
\mathbf{Pr}[E_3] = \mathbf{Pr}[E_3^{<n}] = \mathbf{Pr}[E_3^{n-1} \mid E_3^{<n-1}]\mathbf{Pr}[E_3^{<n-1}] = \prod_{\ell \in [n-1]} \mathbf{Pr}[E_3^\ell \mid E_3^{<\ell}].
$$
For every $\ell \in [n-1]$, and every set $S$ of $\ell-1$ buyers, it holds that if $i^\star \notin S$ and $M$ selects $S$ as the first $\ell-1$ buyers, the probability of selecting buyer $i^\star$ as the $\ell$-th buyer is $1/(n-(\ell-1))$, by the definition of $\pi$ (particularly because buyer $i^\star$ is a buyer picked uniformly at random). Therefore,
$$
\mathbf{Pr}[E_3] = \mathbf{Pr}[E_3^{<n}] = \prod_{\ell \in [n-1]} \left(1 - \frac{1}{n-(\ell-1)}\right) = \prod_{\ell \in [n-1]} \frac{n-\ell}{n-\ell+1} = \frac{1}{n}.
$$
Thus, we obtain the following upper bound on $RM(I_n)$:
\begin{align*}
RM(I_n) & \leq \mathbf{Pr}[E_1](n-1)m\epsilon + \mathbf{Pr}[E_2]\left(\frac{1}{m} + (n-1)m\epsilon\right) + \frac{1}{n}\left(\frac{(n-1)m\epsilon}{2} + \frac{H_m}{m}\right) \\
& \leq \frac{1}{m} + 2(n-1)m\epsilon + \frac{m\epsilon}{2} + \frac{H_m}{mn} \leq 3(n-1)m\epsilon + \frac{H_m}{mn} + \frac{1}{m}.
\end{align*}

This leads us to conclude that
\[
\frac{RM(I_n)}{OR(I_n)} \leq \frac{3nm\epsilon + H_m/mn + 1/m}{H_m/m} = \frac{3nm^2\epsilon + H_m/n + 1}{H_m} \leq \frac{H_m/n + 4}{H_m} = \frac{1}{n} + \frac{4}{H_{2^n}} \in O\left(\frac{1}{n}\right). \tag*{\qed}
\]
\let\qed\relax
\end{proof}

\subsection{Revenue guarantees for blind offer mechanisms}\label{sec:blind}
We prove in this section that blind offer mechanisms can always extract a constant fraction of the optimal revenue, without making any assumptions on the valuation distribution. Specifically we prove the following theorem.
\begin{theorem}\label{mainprop-blindoffer}
For every instance $(n,\pi,k)$, there is a dominant strategy IC blind offer mechanism for which the expected revenue is at least a $(2 - \sqrt{e})/4 \approx 0.088$ fraction of the maximum expected revenue that can be extracted by an ex-post IC, ex-post IR mechanism. Moreover, if $k = n$, then there is a blind offer mechanism for which the expected revenue equals the full maximum expected revenue that can be extracted by an ex-post IC, ex-post IR mechanism. 
\end{theorem}
We need to establish some intermediate results in order to build up to a proof for the above theorem.
First, we derive an upper bound on the revenue of the optimal ex-post IC, ex-post IR mechanism.
For a given instance $(n,\pi,k)$, consider the linear program with variables $(y_i(\vec{v}))_{i \in [n], \vec{v} \in \text{supp}(\pi)}$ where the objective is 
\begin{equation}
\max \sum_{i \in [n]} \sum_{\vec{v}_{-i} \in \text{supp}(\pi_{-i})} \pi_{-i}(\vec{v}_{-i}) \sum_{v_i \in \text{supp}(\pi_{i,\vec{v}_{-i}})} \mathbf{Pr}_{v_i' \sim \pi_{i,\vec{v}_{-i}}}[v_i' \geq v_i] v_i y_i(v_i,\vec{v}_{-i}) \label{lp30}
\end{equation}
subject to the constraints 
\begin{align}
& \qquad \forall i \in [n], \vec{v}_{-i} \in \text{supp}(\pi_{-i}) \colon \sum_{v_i \in \text{supp}(\pi_{i,\vec{v}_{-i}})} y_i(\vec{v}) \leq 1, \label{lp31}\\
& \qquad \forall \vec{v} \in \text{supp}(\pi) \colon \sum_{i \in [n]} \sum_{v_i' \in \text{supp}(\pi_{i,\vec{v}_{-i}}) \colon v_i' \leq v_i} y_i(v_i',\vec{v}_{-i}) \leq k, \label{lp33} \\
& \qquad \forall i \in [n], \vec{v} \in \text{supp}(\pi) \colon y_i(\vec{v}) \geq 0. \label{lp32}
\end{align}
The next lemma states that the solution to this linear program forms an upper bound on the revenue of the optimal mechanism, and that the solution to the above linear program is integral in case $k = n$.

\begin{lemma}\label{thm:lp}
For any instance $(n,\pi,k)$, the linear program (\ref{lp30}-\ref{lp32}) upper bounds the maximum expected revenue achievable by an ex-post IC, ex-post IR mechanism. Moreover, when $k = n$ the optimal solution to (\ref{lp30}--\ref{lp32}) is to set $y_i(v_i,\vec{v}_{-i})$ to $1$ for the value $v_i$ that maximizes $v_i\mathbf{Pr}_{v_i' \sim \pi_{i,\vec{v}_{-i}}}[v_i' \geq v_i]$ (for all $i \in [n], \vec{v} \in \text{supp}(\pi)$).
\end{lemma}
\begin{proof}[Proof sketch.]
Integrality for $k = n$ is the easiest to prove among these two claims, so we do that first. Note that in case $k = n$, we can safely remove the constraints (\ref{lp33}) from the linear program, because when $k=n$ these constraints are implied by (\ref{lp31}) and (\ref{lp32}). The linear program that remains tells us how to optimize a sum of convex combinations of values. That is, it effectively tells us to pick for each $i \in [n]$ and $\vec{v}_{-i} \in \text{supp}(\pi_{-i})$ a convex combination of the values $\{v_i\mathbf{Pr}_{v_i' \sim \pi_{i,\vec{v}_{-i}}}[v_i' \geq v_i]\}_{v_i' \in \text{supp}(\pi_{i,\vec{v}_{-i}})}$. The optimal solution is therefore to put weight $1$ on the maximum values in these sets, i.e., to set $y_i(v_i,\vec{v}_{-i})$ to $1$ for the value $v_i$ that maximizes $v_i\mathbf{Pr}_{v_i' \sim \pi_{i,\vec{v}_{-i}}}[v_i' \geq v_i]$.

To prove the first of the two claims in the lemma, we first prove that a monotonicity constraint holds on the set of possible allocations that a ex-post IC, ex-post IR mechanism can output. Moreover, we show that the prices charged by the mechanism cannot exceed a certain upper bound given in terms of allocation probabilities. Then, we formulate a linear program whose optimal value equals the revenue of the optimal ex-post IC, ex-post IR mechanism. We finally rewrite the latter linear program into (\ref{lp30}--\ref{lp32}). This proof relies on the approach introduced in \cite[Section 4.2]{guptanagarajan}.
\end{proof}

We can now proceed to prove our main result about blind offer mechanisms. We first handle the case of unlimited supply. Consider the following blind offer mechanism.
\begin{definition}\label{optmechanism}
Consider an instance $(n,\pi,n)$. For $i \in [n]$ and $\vec{v}_{-i} \in \text{supp}(\pi_{-i})$, fix $\hat{p}_{i,\vec{v}_{-i}}$ to be any value in the set $\arg_p \max\{p\mathbf{Pr}_{v_i \sim \pi_{i,\vec{v}_{-i}}}[p\leq v_i] \colon p \in \mathbb{R}\}$. Define $M_{\pi}^n$ to be the following blind offer mechanism for allocating service to $n$ buyers when the valuations of these buyers are drawn from $\pi$. 

Let $\vec{b}$ be the submitted vector of bids. For $i \in [n]$, if $\vec{b}_{-i} \in \text{supp}(\pi_{-i})$ and $b_i \geq \hat{p}_{i,\vec{b}_{-i}}$, then $M_{\pi}^{n}$ gives service to $i$ and charges $i$ the price $\hat{p}_{i,\vec{b}_{-i}}$. If $b_i < \hat{p}_{i,\vec{b}_{-i}}$, then the price charged to $i$ is 0, and $i$ is not given service. Otherwise, if $\vec{b}_{-i} \not\in \text{supp}(\pi_{-i})$, the price charged to $i$ is $0$ and $i$ is not allocated service.
\end{definition}

\begin{lemma} \label{optimal-unlimited}
For instance $(n,\pi,n)$, mechanism $M_{\pi}^n$ extracts the maximum revenue among the class of ex-post IC, ex-post IR mechanisms.
\end{lemma}
\begin{proof}
Denote by $p_i(\vec{v})$ the price charged to buyer $i \in [n]$ when the buyers have valuation vector $\vec{v} \in \text{supp}(\pi)$. We can write the expected revenue of $M_{\pi}^n$ as follows:
\begin{align*}
\mathbf{E}_{\vec{v} \sim \pi}\left[\sum_{i\in [n]} p_i(\vec{v})\right] & = \sum_{i \in [n]} \mathbf{E}_{\vec{v} \sim \pi}[p_i(\vec{v})] = \sum_{i \in [n]} \sum_{\vec{v} \in \text{supp}(\pi)} \pi(\vec{v}) p_i(\vec{v}) \\
& = \sum_{i \in [n]} \sum_{\vec{v}_{-i} \in \text{supp}(\pi_{-i})} \pi_{-i}(\vec{v}_{-i}) \mathbf{Pr}_{v_i \sim \pi_{i,\vec{v}_{-i}}}[v_i \geq \hat{p}_{i,\vec{v}_{-i}}] \hat{p}_{i,\vec{v}_{-i}}.
\end{align*}
Lastly, by Lemma \ref{thm:lp} and the objective function (\ref{lp30}) of the linear program, we conclude that the latter expression is equal to the solution of the linear program, which is an upper bound on the optimal revenue among all ex-post IC, ex-post IR mechanisms by Theorem \ref{thm:lp}. 
\end{proof}

For the case of $k$-limited supply where $k < n$, things are somewhat more complicated. Indeed, there does not seem to exist a blind offer mechanism as simple and elegant as $M_{\pi}^n$. However, we are still able to define a blind offer mechanism that extracts a constant fraction of the optimal revenue.
\begin{definition}\label{optmechanismunit}
Let $(n,\pi,k)$ be an arbitrary instance. Let $(y_i^*(\vec{v}))_{i \in [n]}$ be the optimal solution to the linear program (\ref{lp30}--\ref{lp32}) corresponding to this instance. 

Let $M_{\pi}^k$ be the blind offer mechanism that does the following: let $\vec{v}$ be the vector of submitted valuations. Iterate over the set of buyers such that in iteration $i$, buyer $i$ is picked. In iteration $i$, select one of the following options: offer service to buyer $i$ at a price $p$ for which it holds that $y_i^*(p,\vec{b}_{-i}) > 0$, or skip buyer $i$. The probabilities with which these options are chosen are as follows: price $p$ is offered with probability $y_i^*(p,\vec{b}_{-i})/2$, and buyer $i$ is skipped with probability $1 - \sum_{p' \in \text{supp}(\pi_{i,\vec{b}_{-i}})} y_i^*(p,\vec{b}_{-i})/2$. The mechanism terminates once $k$ buyers have accepted an offer, or when iteration $n + 1$ is reached.
\end{definition}
\begin{lemma} 
\label{constantapx-klimited}
On instance $(n,\pi,k)$, the expected revenue of blind offer mechanism $M_{\pi}^k$ is at least a $(2 - \sqrt{e})/4 \approx 0.088$ fraction of the expected revenue of the optimal ex-post IC, ex-post IR mechanism. Moreover, there exists a \emph{dominant strategy IC} blind offer mechanism of which the expected revenue is at least a $(2 - \sqrt{e})/4 \approx 0.088$ fraction of the expected revenue of the optimal ex-post IC, ex-post IR mechanism.
\end{lemma}
\begin{proof}[Proof sketch.]
 We will show that the expected revenue of $M_{\pi}^k$ is at least
\begin{equation*}
\frac{2 - \sqrt{e}}{4} \sum_{i \in [n]} \sum_{\vec{v}_{-i} \in \text{supp}(\pi_{-i})} \pi_{-i}(\vec{v}_{-i}) \sum_{v_i \in \text{supp}(\pi_{i,\vec{v}_{-i}})} \mathbf{Pr}_{v_i' \sim \pi_{i,\vec{v}_{-i}}}[v_i' \geq v_i] v_i y_i^*(v_i,\vec{v}_{-i}),
\end{equation*}
which, by Lemma \ref{thm:lp} and the LP objective function (\ref{lp30}), is a $(2 - \sqrt{e})/4$ fraction of the expected revenue of the optimal ex-post IC, ex-post IR mechanism.

For a vector of valuations $\vec{v} \in \text{supp}(\pi)$ and a buyer $i \in [n]$, denote by $D_{i,\vec{v}_{-i}}$ the probability distribution from which mechanism $M_{\pi}^k(\vec{v})$ draws a price that is offered to buyer $i$, in case iteration $i \in [n]$ is reached (as described in Definition \ref{optmechanismunit}). We let $V$ be a number that exceeds $\max\{v_i \colon i \in [n], \vec{v} \in \text{supp}(\pi)\}$ and represent by $V$ the option where $M_{\pi}^k(\vec{v})$ chooses to skip buyer $i$, so that $D_{i,\vec{v}_{-i}}$ is a probability distribution on the set $\{V\}\cup\{v_i \colon y_i^*(v_i,\vec{v}_{-i}) > 0\}$.
Then,
\begin{equation}
\label{eq:expec_rev}
 \mathbf{E}_{\vec{v} \sim \pi}[\text{revenue of } M_{\pi}^k(\vec{v}) ]
 \geq \sum_{i \in [n]} \sum_{\vec{v} \in \text{supp}(\pi)} \pi(\vec{v}) \sum_{\substack{p_i \in \text{supp}(D_{i,\vec{v}_{-i}}) \\ \colon p_i \leq v_i }} \frac{p_i y_i^*(p_i,\vec{v}_{-i})}{2} \mathbf{Pr}_{\forall_i:p_i \sim D_{i,\vec{v}_{-i}}}[|\{j \in [n-1] \colon p_j \leq v_j\}| < k].
\end{equation}
Then, by applying a Chernoff bound, we can prove that
\begin{equation}
\label{eq:bound_chernoff}
 \mathbf{Pr}_{\forall_i:p_i \sim D_{i,\vec{v}_{-i}}}[|\{j \in [n-1] \colon p_j \leq v_j\}| < k] \geq 1 - \left(\frac{e}{4}\right)^{k/2} \geq 1 - \left(\frac{e}{4}\right)^{1/2} = \frac{2 - \sqrt{e}}{2}.
\end{equation}
The first of the two claims of Theorem \ref{constantapx-klimited} follows by combining \eqref{eq:expec_rev} and \eqref{eq:bound_chernoff}. For the second claim,  note that \eqref{eq:bound_chernoff} gives a lower bound of $(2 - \sqrt{e})/2$ on the probability that all players get selected by the mechanism. Therefore, we can combine \eqref{eq:bound_chernoff} with the principle explained in Remark \ref{rem:blindofferic} that allows us to transform $M_{\pi}^k$ into a dominant strategy IC blind offer mechanism. The second claim follows by observing that \eqref{eq:expec_rev} still holds for the transformed mechanism.
\end{proof}
Theorem~\ref{mainprop-blindoffer} now follows by combining Lemmas \ref{optimal-unlimited} and \ref{constantapx-klimited}.
We note that the approximation factor of the theorem is certainly not tight and can be improved with additional work. For example, it is possible to show that for $k=1$ the revenue of $M_{\pi}^k$ is in fact at least $1/4$ of the optimal revenue. Moreover, recall that mechanism $M_{\pi}^k$ works by scaling the probabilities $y_i(\vec{v})$ down by $1/2$. By making this scaling factor dependent on $k$ and choosing it appropriately, we can improve the approximation factor further. We emphasize that the focus and purpose of the above result is merely to show that a constant factor of the optimal revenue (independent of the supply $k$) is achievable.

\subsection{Revenue guarantees for enhanced sequential posted price mechanisms}\label{sec:generalizationpositive}
Finally, in this section we evaluate the revenue guarantees of the enhanced sequential posted price mechanisms in the presence of a form of limited dependence that we will call \emph{$d$-dimensional dependence}, for $d \in \mathbb{N}$.
These are probability distributions for which it holds that the valuation distribution of a buyer conditioned on the valuations of the rest of the buyers can be retrieved by only looking at the valuations of a certain subset of $d$ buyers.
Formally, we have the following definition.
\begin{definition}
A probability distribution $\pi$ on $\mathbb{R}^n$ is \emph{$d$-dimensionally dependent} iff for all $i \in [n]$ there is a subset $S_i \subseteq [n]\setminus \{i\}, |S_i| = d$, such that for all $\vec{v}_{-i} \in \text{supp}(\pi_{-i})$ it holds that $\pi_{i,\vec{v}_{S_i}} = \pi_{i,\vec{v}_{-i}}$.
\end{definition}
Note that if $d = 0$, then $\pi$ is a product of $n$ independent probability distributions on $\mathbb{R}$.
On the other hand, the set of $(n-1)$-dimensionally dependent probability distributions on $\mathbb{R}^n$
equals the set of all probability distributions on $\mathbb{R}^n$.
This notion is useful in practice for settings where it is expected that a buyer's valuation distribution
has a reasonably close relationship with the valuation of a few other buyers.
As an example of one of these practical settings consider the case that
there is a true valuation $v$ for the item,
an expert buyer that keeps this valuation,
and remaining buyers  whose valuation is influenced by independent noise.
It is then sufficient to know the valuation of a single buyer, namely the expert one,
in order to retrieve the exact conditional distribution of any other buyer.
We would like to stress that in order to make this distribution $1$-dimensional dependent,
it is sufficient that such an expert buyer exists,
even if auctioneer does not know which buyer is the expert. 
Moreover, our
definition of dimensional dependence is quite inclusive; for example, in the
example above, even if each bidder picks another bidder as her own expert and adds noise to 
the valuation of her expert, the distribution would remain $1$-dimensionally 
dependent.   

In general, $d$-dimensional dependence is relevant to many practical settings in which it is not necessary to have complete information about the valuations of all the other buyers in order to say something useful about the valuation of a particular buyer.
This rules out the extreme kind of dependence defined in the proof of Theorem \ref{thm:enhancedlb};
there the distributions are not $(n-2)$-dimensionally dependent,
because for each buyer $i$ it holds that the valuations of all buyers $[n]\setminus\{i\}$
are necessary in order to extract the valuation distribution of $i$ conditioned on the others' valuations.

It is important to realize that the class of $d$-dimensional dependent distributions is a strict superset of the class of \emph{Markov random fields of degree $d$}. A Markov random field of degree $d$ is a popular model to capture the notion of limited dependence, and for that model a more straightforward procedure than the one in the proof below exists for obtaining the same revenue guarantee (see Appendix~\ref{apx:simpleFields}). However, the notion of $d$-dimensional dependence is both more natural (for our setting) and much more general. In fact, we show in Appendix \ref{apx:limdependence} that there exist distributions on $\mathbb{R}^n$ that are $1$-dimensionally dependent, but are not a Markov random field of degree less than $n/2$.

In a sense, our definition of $d$-dimensional dependence resembles the limited dependence condition
under which the \emph{Lovasz Local Lemma} holds.

The main result that we will prove in this section is the following.
\begin{theorem}\label{mainthm-klimited}
For every instance $(n,\pi,k)$ where $\pi$ is $d$ dimensionally dependent, there exists an enhanced sequential posted price mechanism of which the expected revenue is at least a $(2 - \sqrt{e})/(16d) \geq 1/(46d) \in \Omega(1/d)$ fraction of the maximum expected revenue that can be extracted by an ex-post IC, ex-post IR mechanism. Moreover, if $k = n$, then there exists an enhanced posted price mechanism of which the expected revenue is at least a $1/(4d)$ fraction of the maximum expected revenue that can be extracted by an ex-post IC, ex-post IR mechanism.
\end{theorem}

A corollary of this theorem is that the bound of Theorem \ref{thm:enhancedlb} is asymptotically tight.
\begin{corollary}\label{cor:enhancedub}
For every instance $(n,\pi,k)$ there exists an enhanced sequential posted price mechanism of which the expected revenue is at least a $\Omega(1/n)$ fraction of the maximum expected revenue that can be extracted by an ex-post IC, ex-post IR mechanism.
\end{corollary} 

For proving our main result on enhanced sequential posted price mechanisms, we make use of our insights about blind offer mechanisms.
The next lemma shows how we can convert blind offer mechanisms into enhanced sequential posted price mechanisms while losing only a factor of $1/4d$ of the revenue, if the valuation distribution is $d$-dimensionally dependent.
\begin{lemma}\label{blindoffer-to-eppm}
Let $\alpha \in [0,1]$ and let $(n,\pi,k)$ be an instance where $\pi$ is $d$-dimensionally dependent. If there exists a blind offer mechanism that extracts in expectation at least an $\alpha$ fraction of the expected revenue of the optimal dominant strategy IC, ex-post IR mechanism, then there exists an enhanced sequential posted price mechanism that extracts in expectation at least a $\alpha/\max\{4d,1\}$ fraction of the expected revenue of the optimal ex-post IC, ex-post IR mechanism.
\end{lemma}
\begin{proof}
Let $M$ be a blind offer mechanism that extracts in expectation at least an $\alpha$ fraction of the expected revenue of the optimal dominant-strategy IC, ex-post IR mechanism.
Let $p_i^M(\vec{v})$ be the expected price paid to mechanism $M$ by buyer $i \in [n]$ when $\vec{v}$ are the reported valuations. Let $q \in [0,1]$ and consider the following enhanced sequential posted price mechanism $M_q$: the mechanism $M_q$ first partitions $[n]$ into two sets $A$ and $B = [n] \setminus A$. It does so by placing each buyer independently with probability $q$ in set $A$, and placing him in set $B$ otherwise. Then, the mechanism retrieves the vector $\vec{v}_A$ by asking the buyers in $A$ for their valuations. The existence of $M$ implies the existence of a blind offer mechanism $M_B$ that only makes offers to buyers in $B$ such that the expected price $p_i^{M_B}(\vec{v})$ paid to $M_B$ by a buyer in $B$ is at least $p_i^M(\vec{v})$ (this can be achieved by doing the same as $M$, but refraining from offering to buyers in $A$). Mechanism $M_q$ offers each buyer $i \in B$ a price that is determined by simulating $M_B$ as follows: make the same decisions as $M_B$ would, except for that an offer of $M_B$ is skipped if $\pi_{i,\vec{v}_{-i}} \not= \pi_{i,\vec{v}_A}$. 

Let $P$ be the distribution (induced by mechanism $M_q$) on the set of 
partitions of $[n]$ into $2$ sets. For $i \in [n]$, let $S_i \subseteq [n]\setminus\{i\}$ be the set of $d$ buyers such that $\pi_{i,\vec{v}_{S_i}} = \pi_{i,\vec{v}_{-i}}$ for all $\vec{v} \in \mathbb{R}^n$. For $T \subseteq [n]$, let $p_i(T,\vec{v})$ be the expected price paid to $M_q$ by a buyer $i \in T$, conditioned on the event that $B=T$ and $S_i \subseteq A$. Note that $p_i(T,\vec{v}) \geq p_i^{M_B}(\vec{v}) \geq p_i^M(\vec{v})$. Therefore, the expected revenue of $M_q$ is
\begin{align*}
\sum_{\vec{v} \in \text{supp}(\pi)} \pi(\vec{v}) \sum_{i \in [n]} \mathbf{Pr}_{\{A,B\} \sim P}[i \in B \cap S_i \subseteq A]p_i(B,\vec{v}) & = \sum_{\vec{v} \in \text{supp}(\pi)} \pi(\vec{v}) \sum_{i \in [n]} (1-q)q^d p_i(B,\vec{v}) \\
& \geq (1-q)q^d \sum_{\vec{v} \in \text{supp}(\pi)} \pi(\vec{v}) \sum_{i \in [n]}  p_i^M(\vec{v})
\end{align*}
The last (double) summation is at least $\alpha$ times the expected revenue of the optimal dominant strategy IC, ex post IR mechanism, by definition of $M$. Therefore, this mechanism extracts at least a $(1-q)q^d\alpha$ fraction of the optimal revenue. For $d=0$ it is optimal to set $q = 0$, which results in an enhanced sequential posted price mechanism whose revenue is $\alpha$-approximately optimal. For $d = 1$ it is optimal to set $q = 1/2$, which results in a $(\alpha/4)$-approximately optimal enhanced sequential posted price mechanism. For $d \geq 2$ setting $q = 1-1/d$ will achieve the desired approximation ratio, since $\lim_{d \rightarrow \infty}(1-1/d)^d = 1/e$. Moreover, $(1-1/d)^d$ is increasing in $d$, and equals $1/4$ for $d = 2$. 
\end{proof}
We note that in the above proof it is easy to see that we can decrease the fraction $q$ of buyers being probed for their valuation at the cost of worsening the approximation guarantee. 

Theorem \ref{mainthm-klimited} directly follows by combining the lemmas above. 

\section{Open problems}
Besides improving any of the approximation bounds that we established in the present paper, there are many other interesting further research directions.
For example, it would be interesting to investigate the revenue guarantees under the additional constraint that the sequential posted price mechanism be \emph{order-oblivious}: i.e., the mechanism has no control over which buyers to pick, and should perform well for any possible ordering of the buyers.
We are also interested in resolving some questions regarding the use of randomization in our enhanced posted price mechanism that extracts $O(1/d)$ of the optimal revenue: in the current proof it is necessary to pick buyers uniformly at random. Does there exist a deterministic enhanced sequential posted price mechanism that attains the same revenue guarantee, or is randomness a necessity? 

An obvious and interesting research direction is to investigate more general 
auction problems. In particular, to what extent can extended SPP mechanisms 
be applied to auctions having non-identical items? Additionally, can such 
mechanisms be applied to more complex allocation constraints or specific valuation functions for the buyers? The agents may have, for example, a demand of more than one item, or there may be a matroid feasibility constraint on the set of buyers or on the set of items that may be allocated.

\section*{Acknowledgments}
We thank Joanna Drummond, Brendan Lucier, Tim Roughgarden  and anonymous 
referees for their 
constructive comments which we have incorporated. 
This work is supported by the EU FET project MULTIPLEX no. 317532, the ERC StG Project PAAl 259515, the Google Research Award for Economics and Market Algorithms, and the Italian MIUR PRIN 2010-2011 project ARS TechnoMedia -- Algorithmics for Social Technological Networks.

\bibliographystyle{abbrvnat}
\bibliography{thoughts}

\newpage
\appendix
\section*{APPENDIX}

\section{Continuous distribution properties}
\label{apx:cont_prop}
Let $\pi$ be a valuation distribution on $[0,a_i]^n$, with $a_i \in \mathbb{R}_{\geq 0}$ for $i \in [n]$,
with density $f$ that is continuous and nowhere zero.
Distribution $\pi$ is said to satisfy \emph{affiliation} iff for every two valuation vectors $\vec{v},\vec{w} \in \text{supp}(\pi)$
it holds that $f(\vec{v} \wedge \vec{w})f(\vec{v} \vee \vec{w}) \geq f(\vec{v})f(\vec{w})$, where $\vec{v}\wedge \vec{w}$
is the component-wise minimum and $\vec{v} \vee \vec{w}$ is the component-wise maximum.
For $i \in [n]$ and $\vec{v}_{-i} \in \text{supp}(\pi_{-i})$ the \emph{conditional marginal density function} $f_i(\cdot | \vec{v}_{-i})$ is defined as 
\begin{equation*}
f_i(v_i \mid \vec{v}_{-i}) = \frac{f(v_i,\vec{v}_{-i})}{\int_0^a f(t,\vec{v}_{-i})dt},
\end{equation*}
the \emph{conditional revenue curve} $B_i(\cdot \mid \vec{v}_{-i})$ is defined as 
\begin{equation*}
B_i(v_i \mid \vec{v}_{-i}) = v_i\int_{v_i}^a f_i(t \mid \vec{v}_{-i})dt,
\end{equation*}
and the \emph{conditional virtual value} $\phi_i(\cdot \mid \vec{v}_{-i})$ is defined as
\begin{equation*}
\phi_i(v_i \mid \vec{v}_{-i}) = - \frac{\frac{d}{dv_i}B_i(v_i \mid \vec{v}_{-i})}{f_i(v_i \mid \vec{v}_{-i})}.
\end{equation*}
Denote by $F_i(\cdot \mid \vec{v}_{-i})$ the cumulative distribution function corresponding to $f_i(\cdot \mid \vec{v}_{-i})$.
Distribution $\pi$ satisfies \emph{regularity} if $\phi_i(\cdot \mid \vec{v}_{-i})$ is non-decreasing for all $i \in [n]$ and $\vec{v}_{-i} \in \text{supp}(\pi_{-i})$
and it satisfies the \emph{monotone hazard rate} condition if $\frac{1 - F_i(v_i \mid \vec{v}_{-i})}{f_i(v_i \mid \vec{v}_{-i})}$ is non-increasing in $v_i$ for all $i \in [n]$ and $\vec{v}_{-i} \in \text{supp}(\pi_{-i})$.

A discussion and justification for the above notions is outside of the scope of this paper, and we refer the interested reader to \citep{roughgarden-talgamcohen}.

Note that \citet{roughgarden-talgamcohen} proved that for any distribution $\pi$ that satisfies regularity and affiliation
the Myerson mechanism is ex-post IC, ex-post IR and optimal among all ex-post IC and ex-post IR mechanisms.

\section{A sequential posted price mechanism for the $k$-limited setting}
\label{apx:bound_limited}

\begin{definition}
For a valuation distribution $\pi$ on $\mathbb{R}^n$, let $v_{\pi}^{\max(k)}$ and $v_{\pi}^{\min(k)}$ be respectively the maximum $k$th largest and minimum $k$th smallest valuation among the valuation vectors in $\text{supp}(\pi)$. Let $r_{\pi}^{(k)} = v_{\pi}^{\max(k)}/v_{\pi}^{\min(k)}$ be the ratio between these values.
\end{definition}

\begin{theorem}
Let $n \in \mathbb{N}_{\geq 1}$, and let $\pi$ be a probability distribution on $\mathbb{R}^n$.
For any $k \in [n]$, there exists a sequential posted price mechanism that, when run on instance $(n, \pi, k)$,
extracts in expectation at least an $\Omega\left(\frac{1}{\log(r_{\pi}^{(k)})} \cdot \frac{v_\pi^{\max(k)}}{v_\pi^{\max(1)}}\right)$
fraction of the expected revenue of the expected optimal social welfare
(and therefore also of the expected revenue of the optimal dominant strategy IC, ex-post IR mechanism).
\end{theorem}
\begin{proof}
Let $M$ be the sequential posted price mechanism that draws a value $p$ uniformly at random from the set $S = \{v_{\pi}^{\min(k)}2^j \colon j \in [\lceil \log(r_{\pi}^{(k)}) - 1 \rceil] \cup \{0\}\}$. $M$ offers price $p$ to all the buyers in an arbitrary order, until $k$ buyers accept.

Let $\pi_{\max(k)}$ be the probability distribution of the $k$-th highest value of $\pi$.
Note that $|S|$ does not exceed $\log(r_{\pi}^{(k)})$. Therefore the probability that $p$ is the highest possible value (among the values in $S$) that does not exceed the value drawn from $\pi_{\max(k)}$, is equal to $1/\log(r_{\pi}^{(k)})$. More formally, let $\pi_S$ be the probability distribution from which $p$ is drawn; then
\begin{equation*}
\mathbf{Pr}_{v^{\max(k)} \sim \pi_{\max(k)}, p \sim \pi_S }[p \leq v^{\max(k)} \cap (\not\exists p' \in S \colon p' > p \wedge p' \leq v^{\max(k)})] \leq \frac{1}{\log(r_{\pi}^{(k)})}.
\end{equation*}
Thus, with probability $1/\log(r_{\pi}^{(k)})$, the mechanism extracts from each winner a revenue of exactly $v_{\pi}^{\min(k)}2^j$, where $j$ is the number such that the value drawn from $\pi_{\max(k)}$ lies in between $v_{\pi}^{\min(k)}2^j$ and $v_{\pi}^{\min(k)}2^{j+1}$.
This implies that with probability $1/\log(r_{\pi}^{(k)})$ the mechanism extracts from buyer $i$ a revenue that lies a factor of $O\left(\frac{v_\pi^{\max(k)}}{v_\pi^{\max(1)}}\right)$ away from $v_\pi^{\max(1)}$.
This leads to the conclusion that
$$
\mathbf{E}_{\vec{v} \sim \pi}[\text{revenue of }M(\vec{v})] \geq \Omega\left(\frac{1}{\log(r_{\pi}^{(k)})} \cdot \frac{v_\pi^{\max(k)}}{v_\pi^{\max(1)}}\right) \sum_{i \in W_M} v_\pi^{\max(1)},$$
where $W_M$ denotes the set of buyers for which the mechanism $M$ allocates the service.
The theorem then follows since $\sum_{i \in W_M} v_\pi^{\max(1)} = \sum_{i \in W_{OPT}} v_\pi^{\max(1)} \geq OPT = \sum_{i \in W_{OPT}} v_i$,
where $W_{OPT}$ denotes the set of buyers at which the optimal mechanism allocates the service,
and $OPT$ is the social welfare achieved by the optimal mechanism.
\end{proof}

We say that an instance $(n,\pi,k)$ is \emph{$k$-well-separated} if for any $\vec{v} \in \text{supp}(\pi)$
the $k$-th coordinate-wise maximum $v_{\pi}^{\max(k)}$ is achieved only by a single buyer,
i.e., the set $\{ i \colon v_i = v_{\pi}^{\max(k)}, \vec{v} \in \text{supp}(\pi)\}$ is a singleton.
Then we can prove the following proposition.
\begin{proposition}
Let $n \in \mathbb{N}_{\geq 1}$, and let $\pi$ be a discrete probability distribution on $\mathbb{R}^n$.
For any $k \in [n]$, if the instance $(n, \pi, k)$ is \emph{$k$-well-separated},
then there exists a sequential posted price mechanism that, when run on instance $(n, \pi, k)$,
extracts in expectation at least an $\Omega\left(\frac{1}{\log(r_{\pi}^{(k)})} \cdot \max_{i \in [n]} \log \frac{v_\pi^{\max(k)}}{v_{\pi,i}^{\max}}\right)$
fraction of the expected optimal social welfare
(and therefore also of the expected revenue of the optimal dominant strategy IC and ex-post IR mechanism).
\end{proposition}
\begin{proof}
Since $\pi$ is discrete, let $\delta$ be the smallest ratio larger than 1 between two valuation in $\vec{v} \in \text{supp}(\pi)$,
i.e., $\delta = \min_{i,j} \{ v_i/v_j > 1 \colon \vec{v} \in \text{supp}(\pi)\}$.
Consider $\epsilon \leq \delta$ and let $M$ be the sequential posted price mechanism that draws a value $p$
uniformly at random from the set $S = \{v_{\pi}^{\min(k)}\epsilon^j \colon j \in [\lceil \log_\epsilon(r_{\pi}^{(k)}) - 1 \rceil] \cup \{0\}\}$.
Moreover, $M$ draws for each $i \in [n]$ a value $p_i$ uniformly at random from the set $S_i = \{v_{\pi,i}^{\min}\epsilon^\ell \colon \ell \in [\lfloor \log_\epsilon \frac{p}{v_{\pi,i}^{\min}} \rfloor, \lceil \log_\epsilon \frac{p}{v_{\pi,i}^{\min}} -1 \rceil]\}$. $M$ proposes prices to the buyers in an arbitrary order, and offers price $p_i$ to buyer $i$.

Let $\pi_{\max(k)}$ be the probability distribution of the $k$-th coordinate-wise maximum of $\pi$.
Note that $|S|$ does not exceed $\log_\epsilon(r_{\pi}^{(k)})$. Therefore the probability that $p$ is the highest possible value (among the values in $S$) that does not exceed the value drawn from $\pi_{\max(k)}$, is equal to $1/\log_\epsilon(r_{\pi}^{(k)})$. More formally, let $\pi_S$ be the probability distribution from which $p$ is drawn; then
\begin{equation*}
\mathbf{Pr}_{v^{\max(k)} \sim \pi_{\max(k)}, p \sim \pi_S }[p \leq v^{\max(k)} \cap (\not\exists p' \in S \colon p' > p \wedge p' \leq v^{\max(k)})] \leq \frac{1}{\log_\epsilon(r_{\pi}^{(k)})}.
\end{equation*}
Thus, with probability $1/\log_\epsilon(r_{\pi}^{(k)})$, the mechanism selects $p = v_{\pi}^{\min(k)}\epsilon^j$, where $j$ is the number such that the value drawn from $\pi_{\max(k)}$ lies in between $v_{\pi}^{\min(k)}\epsilon^j$ and $v_{\pi}^{\min(k)}\epsilon^{j+1}$.
When this event occurs, since the instance is $k$-well separated and by our choice of $\epsilon$,
the set $W_M$ of buyers whose valuation is at least $p$ has size exactly $k$
and corresponds of the set $W_{OPT}$ of buyers with the $k$ highest valuation in $\vec{v} \in \text{supp}(\pi)$.
Hence, with probability $1/\log_\epsilon(r_{\pi}^{(k)})$ the mechanism $M$ extracts revenue only from buyers in $W_{OPT}$.

Now, for any $i \in W_{OPT}$, let $\pi_{i}$ be the probability distribution of the $i$-th coordinate of $\pi$.
Note that $|S_i|$ does not exceed $\log_\epsilon \frac{v_{\pi,i}^{\max}}{v_\pi^{\max(k)}}$. Therefore the probability that $p_i$ is the highest possible value (among the values in $S_i$) that does not exceed the value drawn from $\pi_{i}$, is at least $\log_\epsilon \frac{v_\pi^{\max(k)}}{v_{\pi,i}^{\max}}$. More formally, let $\pi_{S_i}$ be the uniform distribution on $S$; then
\begin{equation*}
\mathbf{Pr}_{v_{i} \sim \pi_{i}, p_i \sim \pi_{S_i} }[p_i \leq v_{i} \cap (\not\exists p'_i \in S_i \colon p'_i > p_i \wedge p'_i \leq v_{i})] \leq \log_\epsilon \frac{v_\pi^{\max(k)}}{v_{\pi,i}^{\max}}.
\end{equation*}
Thus, with probability $\frac{1}{\log_\epsilon(r_{\pi}^{(k)})} \cdot \log_{\epsilon} \frac{v_\pi^{\max(k)}}{v_{\pi,i}^{\max}}$, the mechanism extracts from buyer $i \in W_{OPT}$ a revenue of exactly $v_{i,\min}2^{\ell}$, where $\ell$ is the number such that the value drawn from $\pi_{i}$ lies in between $v_{i,\min}2^\ell$ and $v_{i,\min}2^{\ell+1}$. This implies that with probability $\frac{1}{\log_\epsilon(r_{\pi}^{(k)})} \cdot \log_{\epsilon} \frac{v_\pi^{\max(k)}}{v_{\pi,i}^{\max}}$ the mechanism extracts from buyer $i \in W_{OPT}$ a revenue that lies a factor of at most $1/2$ away from $v_i$. This leads to the conclusion that 
\begin{align*}
\mathbf{E}_{\vec{v} \sim \pi}[\text{revenue of }M(\vec{v})] & \geq \sum_{i \in W_{OPT}} \frac{1}{\log_\epsilon(r_{\pi}^{(k)})} \cdot \log_{\epsilon} \frac{v_\pi^{\max(k)}}{v_{\pi,i}^{\max}} \frac{1}{2} \mathbf{E}_{v_i \sim \pi_i}[v_i] \\
 & = \Omega\left(\frac{1}{\log(r_{\pi}^{(k)})} \cdot \max_{i \in [n]} \log \frac{v_\pi^{\max(k)}}{v_{\pi,i}^{\max}}\right) \sum_{i \in [n]}\mathbf{E}_{v_i \sim \pi_i}[v_i] \\
 & = \Omega\left(\frac{1}{\log(r_{\pi}^{(k)})} \cdot \max_{i \in [n]} \log \frac{v_\pi^{\max(k)}}{v_{\pi,i}^{\max}}\right) \mathbf{E}_{\vec{v} \sim \pi}[OPT(\vec{v})],
\end{align*} 
where $OPT(\vec{v}) = \sum_{i \in W_{OPT}} v_i$ denotes the optimal social welfare when the buyers have valuation vector $v$.
\end{proof}

\section{Addressing some practical problems of enhanced sequential posted price mechanisms}
\label{apx:espp_fix}
The first problematic aspect is that while there is no incentive for a buyer to lie, there is also no incentive to tell the truth.
Therefore, incentive compatibility is only achieved in weakly dominant strategies.
We note that in the literature many (or perhaps most)  truthful mechanisms 
are only incentive compatible in the weak sense.  
Such mechanisms are of  theoretical interest, and  may 
possibly be turned into more practically satisfactory mechanisms. 

In the case of enhanced SPP mechanisms, the lack of a strong incentive
to be truthful only applies to those buyers who are asked for their 
values, knowing they will not be allocated the item. Such a buyer may not 
cooperate at all, or in stating their value may not be truthful. 
The first problem can be resolved by compensating the buyer with some fixed small amount of money that the auctioneer obtains from the buyers who pay for the service. Having insured some level of cooperation, how do we incentivize
these buyers to be truthful? 

Here is an example of such an adaptation of our enhanced SPP mechanisms that creates the proper strong incentive. 
Suppose now that we have provided an incentive for
 every buyer to
reveal a valuation.
At the start of the auction, using a
cryptographic protocol (or just a normal sealed
envelope),  we ask each of the buyers for
a sealed commitment of their value. 
Furthermore, for buyers
being offered a price, with some
(say small) probability, the buyer
must reveal their private valuation in order to be
allowed the item.
If the revealed value is larger than the offered price,
then the buyer is punished and the item is not allocated to him.
Now this is strongly incentive compatible if
we  assume  buyers are risk
averse so that they will not over-bid their
valuation. There is clearly
no monetary reason for a buyer to under-bid.

\section{On $d$-dimensional dependence versus Markov random fields of degree $d$}
\label{apx:limdependence}
This section is intended for readers who are interested in the relative generality of $d$-dimensionally dependence compared to Markov random fields of degree $d$. We assume that reader is familiar with the definition of Markov random fields.
For convenience we will state a weaker notion here. 
\begin{definition}
Given a undirected graph $G = ([n],E)$, a probability distribution $\pi$ on $\mathbb{R}^n$ is a \emph{local Markov random field with respect to $G$} if the following property, named \emph{local Markov property}, holds: for all $i\in [n]$, $\pi_i$ is independent of $\pi_{[n]\setminus (\{i\} \cup \Gamma(i))}$ when conditioning on all coordinates in $\Gamma(i)$. Here, $\Gamma(i)$ denotes the neighborhood of $i$ in $G$.
\end{definition}
(In a true Markov random field, two additional technical conditions needs to be satisfied, called the \emph{pairwise Markov property} and the \emph{global Markov property}.)
We will give an example of a $1$-dimensionally dependent distribution that is not a local Markov random field with respect to any graph $G$ in which all vertices have a degree less than $(n-2)/2$.

Consider a distribution $\pi$ on $\{0,1\}^{n+2}$. A vector $v$ drawn from $\pi$ is formed according to the following random process: we are given $2n$ distinct probability distributions on $\{0,1\}$. We name these distributions $\pi^{i,0}$ and $\pi^{i,1}$, for $i\in[n]$. These distributions are such that both $0$ and $1$ occur with positive probability. Let $v'$ be a value drawn from yet another distribution $\pi'$ on $\{0,1\}$ where again both $0$ and $1$ have positive probability. The final generated vector is then $(v^{1,v'}, v^{2,v'}, \ldots, v^{n,v'}, v',v')$, where $v^{i,v'}$ is drawn from $\pi^{i,v'}$.

$\pi$ is clearly $1$-dimensionally dependent, since for $i \in [n]$ the conditional marginal distribution $\pi_{i,\vec{v}_{-i}}$ is determined by only the value $v'$, which is the value of the $(n+1)$-th coordinate. Also, the value of the $(n+1)$-th coordinate is entirely determined by the $(n+2)$-th coordinate, and vice versa.

We can also easily see that $\pi$ is not a Markov random field with respect to any graph in which all vertices have a degree less than $n/2$. Let $G$ be a graph such that $\pi$ is a Markov random field. Suppose for contradiction that there exists an $i \in [n]$ for which it holds that $\Gamma(i) \subseteq [n]$. Then the local Markov property would be violated. Therefore, each vertex in $[n]$ is connected to either vertex $n+1$ or $n+2$. Hence, we conclude that either vertex $n+1$ or $n+2$ has at least $n/2$ vertices attached to it. 

\section{A blind-offer mechanism that achieves $O(1/d)$-approximated revenue for Markov random fields of degree $d$}
\label{apx:simpleFields}
Consider the following marking procedure: given Markov random fields of degree $d$ (in which nodes represent buyers),
at each round, mark a uniformly random node, then remove the node and all its neighbors and repeat until no node is left.
Note that, each node is marked with probability at least $1/(d+1)$.
Consider now the following blind offer mechanism:
do value queries to the unmarked buyers, and run posted pricing on the marked buyers.
The marked buyers' values are independent conditioned on the other values of the others.
Then, one can use the results of \cite{chawla2014approximate} and \cite{yan}
to extract a fraction $1-1/e$ of the revenue extractable from them\footnote{We acknowledge an anonymous reviewer for pointing out this simple algorithm}.

\section{Missing proofs}\label{sec:missingproofs}
\subsection{Proof of Proposition~\ref{prop:logunlimited}}
\label{sec:logunlimited}
\begin{proof}
Let $M$ be the sequential posted price mechanism that draws for each $i \in [n]$ a value $p_i$ uniformly at random from the set $S_i = \{v_{\pi,i}^{\min}2^k \colon k \in [\lceil \log(r_{\pi,i})-1 \rceil] \cup \{0\}\}$. $M$ proposes prices to the buyers in an arbitrary order, and offers price $p_i$ to buyer $i$.

For $i \in [n]$, let $\pi_{i}$ be the probability distribution of the $i$th coordinate of $\pi$.
Note that $|S_i|$ does not exceed $\log(r_{\pi,i})$. Therefore the probability that $p_i$ is the highest possible value (among the values in $S_i$) that does not exceed the value drawn from $\pi_{i}$, is at least $1/\log(r_{\pi,i})$. More formally, let $\pi_{S_i}$ be the uniform distribution on $S$; then
\begin{equation*}
\mathbf{Pr}_{v_{i} \sim \pi_{i}, p_i \sim \pi_{S_i} }[p_i \leq v_{i} \cap (\not\exists p'_i \in S_i \colon p'_i > p_i \wedge p'_i \leq v_{i})] \leq \frac{1}{\log(r_{\pi,i})}.
\end{equation*}
Thus, with probability $1/\log(r_{\pi,i})$, the mechanism extracts from buyer $i$ a revenue of exactly $v_{i,\min}2^{k}$, where $k$ is the number such that the value drawn from $\pi_{i}$ lies in between $v_{i,\min}2^i$ and $v_{i,\min}2^{i+1}$. This implies that with probability $1/\log(r_{\pi,i})$ the mechanism extracts from buyer $i$ a revenue that lies a factor of at most $1/2$ away from $v_i$. This leads to the conclusion that 
\begin{align*}
\mathbf{E}_{\vec{v} \sim \pi}[\text{revenue of }M(\vec{v})] & \geq \sum_{i \in [n]} \frac{1}{\log(r_{\pi,i})}\frac{1}{2}\mathbf{E}_{v_i \sim \pi_i}[v_i] \\
 & \geq \frac{1}{2\log(\max\{r_{\pi,i} \colon i \in [n]\})}\sum_{i \in [n]}\mathbf{E}_{v_i \sim \pi_i}[v_i] \\
 & = \frac{1}{2\log(\max\{r_{\pi,i} \colon i \in [n]\})}\mathbf{E}_{\vec{v} \sim \pi}[OPT(\vec{v})],
\end{align*} 
where $OPT(\vec{v}) = \sum_{i \in [n]} v_i$ denotes the optimal social welfare when the buyers have valuation vector $\vec{v}$.
\end{proof}

\subsection{Proof of Lemma~\ref{thm:lp}}
\begin{proof}
The second claim has been already proved in the proof sketch.
It remains to prove the first claim.
To this aim, let us first introduce some specialized notation:
Let $\sigma$ now be a probability distribution on a finite subset of $\mathbb{R}_{\geq 0}$
and $x \in \text{supp}(\sigma)$,
we write $\text{prec}_{\sigma}(x)$ to denote $\max \text{supp}(\sigma) \cap [0, x)$ if $\text{supp}(\sigma) \cap [0, x)$ is non-empty.
Otherwise, if $\text{supp}(\sigma) \cap [0, x) = \varnothing$, we define $\text{prec}_{\sigma}(x) = 0$.
Similarly, we write $\text{succ}_{\sigma}(x)$ to denote $\min \text{supp}(\sigma) \cap (x, \infty]$.
(We leave $\text{succ}_{\sigma}(x)$ undefined if $\text{supp}(\sigma) \cap (x, \infty]$ is empty.)

Suppose now that $A$ is an optimal dominant strategy IC, ex-post IR mechanism. For $\vec{v} \in \text{supp}(\pi)$, denote by $x(\vec{v})$ the expected allocation vector output by $A$ when the buyers report valuation vector $\vec{v}$ (so that for $i \in [n]$, the value $x_i(\vec{v})$ is the probability that $i$ gets allocated service, when the buyers report $\vec{v}$) and denote by $p(\vec{v})$ the vector of expected prices charged by $A$ when the buyers report $\vec{v}$. 
Ex-post incentive compatibility states that 
\begin{equation*}
\forall i \in [n], \vec{v}_{-i} \in \text{supp}(\pi_{-i}), (v_i,v_i') \in \text{supp}(\pi_i)^2 \colon v_ix_i(v_i,\vec{v}_{-i}) - p_i(v_i,\vec{v}_{-i}) \geq v_ix_i(v_i',\vec{v}_{-i}) - p_i(v_i',\vec{v}_{-i}),
\end{equation*}
and ex-post individual rationality states that
\begin{equation*}
\forall i \in [n], \vec{v} \in \text{supp}(\pi) \colon v_ix_i(v_i,\vec{v}_{-i}) - p_i(v_i,\vec{v}_{-i}) \geq 0.
\end{equation*}

The next lemma states that, in $A$, the allocation probability for a buyer is non-decreasing in his reported valuation.
\begin{lemma}
\label{monotonicity}
For all $i \in [n]$, all $\vec{v}_{-i} \in \pi_{-i}$, and all $v_i, v_i' \in \text{supp}(\pi_{i,\vec{v}_{-i}})$, with $v_i < v_i'$, it holds that 
\begin{equation*}
x_i(v_i, \vec{v}_{-i}) \leq x_i(v_i', \vec{v}_{-i}).
\end{equation*}
\end{lemma}
\begin{proof}
By way of contradiction, we assume that $\epsilon = x_i(v_i,\vec{v}_{-i}) - x_i(v_i',\vec{v}_{-i}) > 0$. 
By ex-post incentive compatibility it holds that
\begin{align*}
v_ix_i(v_i,\vec{v}_{-i}) - p_i(v_i,\vec{v}_{-i}) & \geq v_ix_i(v_i',\vec{v}_{-i}) - p_i(v_i',\vec{v}_{-i}), \\
v_i'x_i(v_i',\vec{v}_{-i}) - p_i(v_i',\vec{v}_{-i}) & \geq v_i'x_i(v_i,\vec{v}_{-i}) - p_i(v_i,\vec{v}_{-i}).
\end{align*}
We now rewrite these inequalities as
\begin{align*}
v_ix_i(v_i,\vec{v}_{-i}) - v_ix_i(v_i',\vec{v}_{-i}) & \geq p_i(v_i,\vec{v}_{-i}) - p_i(v_i', \vec{v}_{-i}), \\
v_i'x_i(v_i,\vec{v}_{-i}) - v_i'x_i(v_i',\vec{v}_{-i}) & \leq p_i(v_i,\vec{v}_{-i}) - p_i(v_i', \vec{v}_{-i}).
\end{align*}
This results in the following pair of inequalities.
\begin{align*}
v_i\epsilon & \geq p_i(v_i,\vec{v}_{-i}) - p_i(v_i', \vec{v}_{-i}), \\
v_i'\epsilon & \leq p_i(v_i,\vec{v}_{-i}) - p_i(v_i', \vec{v}_{-i}).
\end{align*}
The two inequalities contradict each other, because $v_i' > v_i$ and we assumed $\epsilon > 0$. 
\end{proof}

The next lemma upper bounds the prices charged by $A$.
\begin{lemma}\label{lem:priceub}
For all $i \in [n]$, all $\vec{v}_{-i} \in \text{supp}(\vec{v}_{-i})$ and all $v_i \in \text{supp}(\pi_{i,\vec{v}_{-i}})$, it holds that
\begin{equation}\label{priceub}
p_{i}(v_i, \vec{v}_{-i}) \leq v_i x_i(v_i,\vec{v}_{-i})-\sum_{v_i' \in \text{supp}(\pi_{i,\vec{v}_{-i}}) \colon v_i' < v_i} (\text{succ}_{\pi_{i,\vec{v}_{-i}}}(v_i) - v_i)x_i(v_i',\vec{v}_{-i}).
\end{equation}
\end{lemma}
\begin{proof}
For $v_i' \in \text{supp}(\pi_{i,\vec{v}_{-i}})$ the ex-post IC constraint for $\vec{v}_{-i}, v_i', \text{prec}_{\pi_{i,\vec{v}_{-i}}}(v_i)$ can be written as
\begin{equation*}
v_i'(x_i(v_i',\vec{v}_{-i}) - x_i(v_i'',\vec{v}_{-i})) \geq p_i(v_i',\vec{v}_{-i}) - p_i(v_i'',\vec{v}_{-i}),
\end{equation*}
where $v_i'' = \text{prec}_{\pi_{i,\vec{v}_{-i}}}(v_i)$.
Summing the above over all $v_i' \in \text{supp}(\pi_{i,\vec{v}_{-i}}), v_i' < v_i$ yields (\ref{priceub}).
\end{proof}

The optimal revenue among all ex-post IC, ex-post IR mechanisms (and thus the expected revenue of $A$) can be written as the following linear program, where $(x(\vec{v}))_{\vec{v} \in \text{supp}(\pi)}$ and $(p(\vec{v}))_{\vec{v} \in \text{supp}(\pi)}$ are the variables:
\begin{align}
& \max \Bigg\{ \sum_{i \in [n]} \sum_{\vec{v} \in \text{supp}(\pi)} p_i(\vec{v}) \qquad \Bigg| \\
& \qquad \forall i \in [n], \vec{v} \in \text{supp}(\pi) \colon v_ix_i(\vec{v}) - p_i(\vec{v}) \geq 0 \label{eq:lpopt1} \\
& \qquad \forall i \in [n], (v_i,v_i') \in \text{supp}(\pi_i)^2, \vec{v}_{-i} \in \text{supp}(\pi_{-i,v_i}) \colon v_ix_i(v_i,\vec{v}_{-i}) - p_i(v_i,\vec{v}_{-i}) \geq v_ix_i(v_i',\vec{v}_{-i}) - p_i(v_i',\vec{v}_{-i}) \label{eq:lpopt2} \\
& \qquad \forall \vec{v} \in \text{supp}(\pi) \colon \sum_{i} x_i(\vec{v}) \leq k \label{eq:lpopt3} \\
& \qquad \forall i \in [n], \vec{v} \in \text{supp}(\pi) \colon 0 \leq x_i(\vec{v}) \leq 1 \Bigg\} \label{eq:lpopt4}
\end{align}
In the above linear program, (\ref{eq:lpopt1}) are the ex-post IR constraints, (\ref{eq:lpopt2}) are the ex-post IC constraints, and (\ref{eq:lpopt3}) expresses that the service cannot be provided to more than $k$ buyers.

By Lemma \ref{monotonicity}, it is possible to add to the above linear program the constraints $x_i(v_i,\vec{v}_{-i}) \geq x_i(\text{prec}_{\pi_{i,\vec{v}_{-i}}}(v_i))$ for $i \in [n], \vec{v}_{-i} \in \text{supp}(\pi_{-i}), v_i \in \text{supp}(\pi_{i,\vec{v}_{-i}})$. Moreover, by Lemma \ref{lem:priceub}, replacing the objective function by 
\begin{equation*}
\sum_{i \in [n]} \sum_{\vec{v} \in \text{supp}(\pi)} \pi(\vec{v}) \left(v_i x_i(v_i,\vec{v}_{-i})-\sum_{v_i' \in \text{supp}(\pi_{i,\vec{v}_{-i}}) \colon v_i' < v_i} (\text{succ}_{\pi_{i,\vec{v}_{-i}}}(v_i) - v_i)x_i(v_i',\vec{v}_{-i})\right)
\end{equation*}
and removing the constraints (\ref{eq:lpopt1}) and (\ref{eq:lpopt2}) results in the following linear program that upper bounds the optimal revenue among the ex-post IC, ex-post IR mechanisms:
\begin{align}
& \max \Bigg\{ \sum_{i \in [n]} \sum_{\vec{v} \in \text{supp}(\pi)} \pi(\vec{v}) \left(v_i x_i(v_i,\vec{v}_{-i})-\sum_{v_i' \in \text{supp}(\pi_{i,\vec{v}_{-i}}) \colon v_i' < v_i} (\text{succ}_{\pi_{i,\vec{v}_{-i}}}(v_i) - v_i)x_i(v_i',\vec{v}_{-i})\right) \qquad \Bigg| \\
& \qquad \forall i \in [n], \vec{v} \in \text{supp}(\pi) \colon x_i(\vec{v}) \geq x_i(\text{prec}_{\pi_{i,\vec{v}_{-i}}}(v_i),\vec{v}_{-i}) \label{lp21} \\
& \qquad \forall \vec{v} \in \text{supp}(\pi) \colon \sum_{i} x_i(\vec{v}) \leq k \label{lp22} \\
& \qquad \forall i \in [n], \vec{v} \in \text{supp}(\pi) \colon 0 \leq x_i(\vec{v}) \leq 1 \Bigg\} \label{lp23}
\end{align}

We will show that the linear program (\ref{lp30}--\ref{lp32}) is equivalent to the above.
Set $y_i(\vec{v}) = x_i(\vec{v}) - x_i(\text{prec}_{\pi_{i,\vec{v}_{-i}}}(v_i),\vec{v}_{-i})$ for all $i \in [n], \vec{v} \in \text{supp}(\pi)$, and observe that the constraints (\ref{lp21}), (\ref{lp22}) and (\ref{lp23}) are then equivalent to (\ref{lp31}), (\ref{lp33}) and (\ref{lp32}) respectively.
Moreover, with this correspondence between $y$ and $x$, we can rewrite the objective function as follows:
\begin{align*}
 &  \sum_{i \in [n]} \sum_{\vec{v} \in \text{supp}(\pi)} \pi(\vec{v}) \left(v_i x_i(v_i,\vec{v}_{-i})-\sum_{v_i' \in \text{supp}(\pi_{i,\vec{v}_{-i}}) \colon v_i' < v_i} (\text{succ}_{\pi_{i,\vec{v}_{-i}}}(v_i) - v_i)x_i(v_i',\vec{v}_{-i})\right)  \\
 & = \sum_{i \in [n]} \sum_{\vec{v}_{-i} \in \text{supp}(\pi_{-i})} \pi_{-i}(\vec{v}_{-i}) \sum_{v_i \in \text{supp}(\pi_{i,\vec{v}_{-i}})} \pi_{i,\vec{v}_{-i}}(v_i) v_i x_i(v_i,\vec{v}_{-i}) \\
 & \quad - \sum_{i \in [n]} \sum_{\vec{v}_{-i} \in \text{supp}(\pi_{-i})} \pi_{-i}(\vec{v}_{-i}) \sum_{v_i \in \text{supp}(\pi_{i,\vec{v}_{-i}})} \pi_{i,\vec{v}_{-i}}(v_i) \sum_{v_i' \in \text{supp}(\pi_{i,\vec{v}_{-i}}) \colon v_i' < v_i} (\text{succ}_{\pi_{i,\vec{v}_{-i}}}(v_i) - v_i)x_i(v_i',\vec{v}_{-i})\\
 & = \sum_{i \in [n]} \sum_{\vec{v}_{-i} \in \text{supp}(\pi_{-i})} \pi_{-i}(\vec{v}_{-i}) \sum_{v_i \in \text{supp}(\pi_{i,\vec{v}_{-i}})} \pi_{i,\vec{v}_{-i}}(v_i) v_i x_i(v_i,\vec{v}_{-i}) \\
 & \quad - \sum_{i \in [n]} \sum_{\vec{v}_{-i} \in \text{supp}(\pi_{-i})} \pi_{-i}(\vec{v}_{-i}) \sum_{v_i \in \text{supp}(\pi_{i,\vec{v}_{-i}})} x_i(v_i,\vec{v}_{-i})(\text{succ}_{\pi_{i,\vec{v}_{-i}}}(v_i) - v_i)\mathbf{Pr}_{v_i' \sim \pi_{i,\vec{v}_{-i}}}[v_i' > v_i] \\
 & = \sum_{i \in [n]} \sum_{\vec{v}_{-i} \in \text{supp}(\pi_{-i})} \pi_{-i}(\vec{v}_{-i}) \sum_{v_i \in \text{supp}(\pi_{i,\vec{v}_{-i}})} x_i(v_i,\vec{v}_{-i}) ( \pi_{i,\vec{v}_{-i}}(v_i) v_i - (\text{succ}_{\pi_{i,\vec{v}_{-i}}}(v_i) - v_i)\mathbf{Pr}_{v_i' \sim \pi_{i,\vec{v}_{-i}}}[v_i' > v_i]) \\
 & = \sum_{i \in [n]} \sum_{\vec{v}_{-i} \in \text{supp}(\pi_{-i})} \pi_{-i}(\vec{v}_{-i}) \sum_{v_i \in \text{supp}(\pi_{i,\vec{v}_{-i}})} x_i(v_i,\vec{v}_{-i}) ( \mathbf{Pr}_{v_i' \sim \pi_{i,\vec{v}_{-i}}}[v_i' \geq v_i] v_i - (\text{succ}_{\pi_{i,\vec{v}_{-i}}}(v_i))\mathbf{Pr}_{v_i' \sim \pi_{i,\vec{v}_{-i}}}[v_i' > v_i]) \\
 & = \sum_{i \in [n]} \sum_{\vec{v}_{-i} \in \text{supp}(\pi_{-i})} \pi_{-i}(\vec{v}_{-i}) \sum_{v_i \in \text{supp}(\pi_{i,\vec{v}_{-i}})} (x_i(v_i,\vec{v}_{-i}) - x_i(\text{prec}_{\pi_{i,\vec{v}_{-i}}}(v_i),\vec{v}_{-i}) ) v_i \mathbf{Pr}_{v_i' \sim \pi_{i,\vec{v}_{-i}}}[v_i' \geq v_i] \\ 
 & = \sum_{i \in [n]} \sum_{\vec{v}_{-i} \in \text{supp}(\pi_{-i})} \pi_{-i}(\vec{v}_{-i}) \sum_{v_i \in \text{supp}(\pi_{i,\vec{v}_{-i}})} \mathbf{Pr}_{v_i' \sim \pi_{i,\vec{v}_{-i}}}[v_i' \geq v_i] v_i y_i(v_i,\vec{v}_{-i}). 
\end{align*}
This completes the proof.
\end{proof}

\subsection{Proof of Lemma~\ref{constantapx-klimited}}
\begin{proof}
We will show that the expected revenue of $M_{\pi}^k$ is at least
\begin{equation*}
\frac{2 - \sqrt{e}}{4} \sum_{i \in [n]} \sum_{\vec{v}_{-i} \in \text{supp}(\pi_{-i})} \pi_{-i}(\vec{v}_{-i}) \sum_{v_i \in \text{supp}(\pi_{i,\vec{v}_{-i}})} \mathbf{Pr}_{v_i' \sim \pi_{i,\vec{v}_{-i}}}[v_i' \geq v_i] v_i y_i^*(v_i,\vec{v}_{-i}),
\end{equation*}
which is by Theorem \ref{thm:lp} and the objective function (\ref{lp30}) a $(2 - \sqrt{e})/4$ fraction of the expected revenue of the optimal ex-post IC, ex-post IR mechanism.

For a vector of valuations $\vec{v} \in \text{supp}(\pi)$ and a buyer $i \in [n]$, denote by $D_{i,\vec{v}_{-i}}$ the probability distribution from which mechanism $M_{\pi}^k(\vec{v})$ draws a price that is offered to buyer $i$, in case iteration $i \in [n]$ is reached (as described in Definition \ref{optmechanismunit}). We let $V$ be a number that exceeds $\max\{v_i \colon i \in [n], \vec{v} \in \text{supp}(\pi)\}$ and represent by $V$ the option where $M_{\pi}^k(\vec{v})$ chooses the ``do nothing''-option during iteration $i$, so that $D_{i,\vec{v}_{-i}}$ is a probability distribution on the set $\{V\}\cup\{v_i \colon y_i^*(v_i,\vec{v}_{-i}) > 0\}$.
Let us formulate an initial lower bound on the expected revenue of $M_{\pi}^k$.
\begin{equation}
 \label{finalsummation}
\begin{aligned}
& \mathbf{E}_{\vec{v} \sim \pi}[\text{revenue of } M_{\pi}^k(\vec{v}) ]\\
& = \mathbf{E}_{\substack{ \vec{v} \sim \pi, \\ p_1 \sim D_{1,\vec{v}_{-1}}, \\ \vdots \\ p_n \sim D_{n,\vec{v}_{-n}}}}\left[\sum_{i \in [n]} p_i\mathbf{1}[p_i \leq v_i]\mathbf{1}[|\{j \in [i-1] \colon p_j \leq v_j\}| < k]\right]\\
& = \sum_{i \in [n]} \mathbf{E}_{\substack{ \vec{v} \sim \pi, \\ p_1 \sim D_{1,\vec{v}_{-1}}, \\ \vdots \\ p_n \sim D_{n,\vec{v}_{-n}}}}\left[p_i\mathbf{1}[p_i \leq v_i]\mathbf{1}[|\{j \in [i-1] \colon p_j \leq v_j\}| < k]\right]\\
& = \sum_{i \in [n]} \sum_{\vec{v} \in \text{supp}(\pi)} \pi(\vec{v}) \sum_{\substack{p_1 \in \text{supp}(D_{1,\vec{v}_{-1}}) \\ \vdots \\ p_i \in \text{supp}(D_{i,\vec{v}_{-i}})}} p_i \mathbf{1}[p_i \leq v_i]\mathbf{1}[|\{j \in [i-1] \colon p_j \leq v_j\}| < k] \prod_{j \in [i]} D_{j,\vec{v}_{-j}}(p_j)\\
& = \sum_{i \in [n]} \sum_{\vec{v} \in \text{supp}(\pi)} \pi(\vec{v}) \sum_{\substack{p_i \in \text{supp}(D_{i,\vec{v}_{-i}}) \\ \colon p_i \leq v_i }} p_i D_{i,\vec{v}_{-i}}(p_i) \sum_{\substack{p_1 \in \text{supp}(D_{1,\vec{v}_{-1}}) \\ \vdots \\ p_{i-1} \in \text{supp}(D_{i-1,\vec{v}_{-(i-1)}}) \\ \colon |\{j \in [i-1] | p_j \leq v_j\}| < k}} \prod_{j \in [i-1]} D_{j,\vec{v}_{-j}}(p_j)\\
& = \sum_{i \in [n]} \sum_{\vec{v} \in \text{supp}(\pi)} \pi(\vec{v}) \sum_{\substack{p_i \in \text{supp}(D_{i,\vec{v}_{-i}}) \\ \colon p_i \leq v_i }}  \frac{p_i y_i^*(p_i,\vec{v}_{-i})}{2} \mathbf{Pr}_{\substack{p_1 \sim D_{1,\vec{v}_{-1}} \\ \vdots \\ p_{i-1} \sim D_{i-1,\vec{v}_{-(i-1)}}}}[|\{j \in [i-1] \colon p_j \leq v_j\}| < k]\\
& \geq \sum_{i \in [n]} \sum_{\vec{v} \in \text{supp}(\pi)} \pi(\vec{v}) \sum_{\substack{p_i \in \text{supp}(D_{i,\vec{v}_{-i}}) \\ \colon p_i \leq v_i }} \frac{p_i y_i^*(p_i,\vec{v}_{-i})}{2} \mathbf{Pr}_{\substack{p_1 \sim D_{1,\vec{v}_{-1}} \\ \vdots \\ p_{n-1} \sim D_{n-1,\vec{v}_{-(n-1)}}}}[|\{j \in [n-1] \colon p_j \leq v_j\}| < k].
\end{aligned}
\end{equation}
For the second equality, we applied linearity of expectation; the third equality follows from the definition of expected value; to obtain the fourth inequality we eliminate the indicator functions by removing the appropriate terms from the summation; in the fifth inequality we substitute $D_{i,\vec{v}_{-i}}(p_i,\vec{v}_{-i})$ and $D_{j,\vec{v}_{-j}}(p_j',\vec{v}_{-j})$ by concrete probabilities. For the last inequality we lower bounded the last probability in the expression by replacing $i$ by $n$.

For $\vec{v} \in \text{supp}(\pi)$ and $i \in [n-1]$, let us denote by $z_i^{\vec{v}}$ the probability that a price drawn from $D_{i,\vec{v}_{-i}}$ does not exceed $v_i$, i.e.,
\begin{equation*}
z_i^{\vec{v}} = \sum_{p_i \in \text{supp}(D_{i,\vec{v}_{-i}}) \colon p_i \leq v_i} D_{i,\vec{v}_{-i}}(p_i) = \sum_{v_i' \in \text{supp}(\pi_{i,\vec{v}_{-i}}) \colon v_i' \leq v_i} \frac{y_i^*(v_i', \vec{v}_{-i})}{2},
\end{equation*} 
 and let $X_i^{\vec{v}}$ denote the random variable that takes the value $1$ with probability $z_i^{\vec{v}}$ and the value $0$ with probability $1- z_i^{\vec{v}}$. Then the final probability in the derivation above, i.e.,
\begin{equation*}
\mathbf{Pr}_{\substack{p_1 \sim D_{1,\vec{v}_{-1}} \\ \vdots \\ p_{n-1} \sim D_{i-1,\vec{v}_{-(n-1)}}}}[|\{j \in [n-1] \colon p_j \leq v_j\}| < k]
\end{equation*}
 can be written as
\begin{equation*}
1 - \mathbf{Pr}\left[\sum_{i \in [n-1]} X_i^{\vec{v}} \geq k \right].
\end{equation*}
Let $\mu = \mathbf{E}\left[\sum_{i \in [n-1]} X_i^{\vec{v}}\right]$. Next, we use a Chernoff bound:
\begin{theorem}[Chernoff bound (as in \citep{motwaniraghavan}).]
Let $X_1, \ldots, X_n$ be independent random $(0,1)$-variables such that, for $i \in [n]$, $\mathbf{Pr}[X_i = 1] = p_i$ where $p_i \in [0,1]$. Then, for $X = \sum_{i \in [n]} X_i$, $\mu = \mathbf{E}[X] = \sum_{i \in [n]} p_i$ and any $\delta > 0$,
\begin{equation*}
\mathbf{Pr}[X \geq (1 + \delta)\mu] \leq \left(\frac{e^{\delta}}{(1+\delta)^{1 + \delta}}\right)^{\mu}.
\end{equation*}
\end{theorem}
This implies that the expression above is bounded as follows. 
$$
1 - \mathbf{Pr}\left[\sum_{i \in [n-1]} X_i^{\vec{v}} \geq \left(1 + \left(\frac{k}{\mu} - 1\right)\right)\mu \right] \geq 1 - \left(\frac{e^{k/\mu - 1}}{(k/\mu)^{k/\mu}}\right)^{\mu}.
$$
By the definition of $z_i^{\vec{v}}$ and constraint (\ref{lp33}) of the linear program, it holds that $\mu = \sum_{i \in [n-1]} z_i^{\vec{v}} \leq k/2$. We can lower bound the expression above by replacing $\mu$ by $k/2$. To see this, rewrite it first into the following:
\begin{equation*}
1 - \left(\frac{e^{k/\mu - 1}}{(k/\mu)^{k/\mu}}\right)^{\mu} = 1 - \frac{e^{k - \mu + k \ln(\mu)}}{k^k}.
\end{equation*}
The derivative of the exponent of $e$ (with respect to $\mu$) is positive for $\mu \in [0,k]$, which means that the exponent of $e$ is increasing in $\mu$ on $[0,k]$. Thus, replacing $\mu$ by its upper bound $k/2$ increases the exponent and therefore decreases the expression above. Therefore:
\begin{equation}\label{eq:bound_chernoff2}
1 - \mathbf{Pr}\left[\sum_{i \in [n-1]} X_i^{\vec{v}} \geq k\right] \geq 1 - \left(\frac{e}{4}\right)^{k/2} \geq 1 - \left(\frac{e}{4}\right)^{1/2} = \frac{2 - \sqrt{e}}{2}.
\end{equation}

Continuing from \eqref{finalsummation}, we obtain
\begin{align*}
& \mathbf{E}_{\vec{v} \sim \pi}[\text{revenue of } M_{\pi}^k(\vec{v}) ] \geq \sum_{i \in [n]} \sum_{\vec{v} \in \text{supp}(\pi)} \pi(\vec{v}) \sum_{p_i \in \text{supp}(D_{i,\vec{v}_{-i}}) \colon p_i \leq v_i } p_i \frac{y_i^*(p_i,\vec{v}_{-i})}{2} \frac{2 - \sqrt{e}}{2} \\
& \qquad = \frac{2 - \sqrt{e}}{4} \sum_{i \in [n]} \sum_{\vec{v} \in \text{supp}(\pi)} \pi(\vec{v}) \sum_{p_i \in \text{supp}(D_{i,\vec{v}_{-i}}) \colon p_i \leq v_i } p_i y_i^*(p_i,\vec{v}_{-i}) \\
& \qquad = \frac{2 - \sqrt{e}}{4} \sum_{i \in [n]} \sum_{\vec{v}_{-i} \in \text{supp}(\pi_{-i})} \pi_{-i}(\vec{v}_{-i}) \sum_{v_i, p_i \in \text{supp}(\pi_{i,\vec{v}_{-i}}) \colon p_i \leq v_i } \pi_{i,\vec{v}_{-i}}(v_i) p_i y_i^*(p_i,\vec{v}_{-i}) \\
& \qquad = \frac{2 - \sqrt{e}}{4} \sum_{i \in [n]} \sum_{\vec{v}_{-i} \in \text{supp}(\pi_{-i})} \pi_{-i}(\vec{v}_{-i}) \sum_{p_i \in \text{supp}(\pi_{i,\vec{v}_{-i}}) \colon p_i \leq v_i } \mathbf{Pr}_{\vec{v}_{-i} \sim \pi_{i,\vec{v}_{-i}}}[v_i \geq p_i] p_i y_i^*(p_i,\vec{v}_{-i}),
\end{align*}
which proves the first of the two claims. For the second claim, observe that \eqref{eq:bound_chernoff2} states a lower bound of $(2 - \sqrt{e})/2$ on the probability that all players get selected by the mechanism. Therefore, we can combine \eqref{eq:bound_chernoff2} with the principle explained in Remark \ref{rem:blindofferic}, which allows us to transform $M_{\pi}^k$ into a dominant strategy IC blind offer mechanism $\hat{M}_{\pi}^k$. Observe now that \eqref{finalsummation} is still a lower bound on the revenue of $\hat{M}_{\pi}^k$ so that the revenue analysis of $M_{\pi}^k$ also holds for $\hat{M}_{\pi}^k$.
\end{proof}

\end{document}